\documentclass[sort&compress, nonatbib]{elsarticle}
\usepackage{mysty}

\begin{document}

\journal{Theoretical Computer Science}

\begin{frontmatter}
  
\title{Closure and Decision Properties \\ for Higher-Dimensional Automata}

\author[e]{Amazigh Amrane}
\author[e]{Hugo Bazille}
\author[e]{Uli Fahrenberg}
\author[w]{Krzysztof Ziemiański}


\address[e]{EPITA Research Laboratory (LRE), France}
\address[w]{University of Warsaw, Poland}


\begin{abstract}
  We report some further developments regarding the language theory of higher-dimensional automata (HDAs).
  Regular languages of HDAs are sets of finite interval partially ordered multisets (pomsets) with interfaces.
  We show a pumping lemma which allows us to expose a class of non-regular languages.
  Concerning decision and closure properties, we show that inclusion of regular languages is decidable (hence is emptiness),
  and that intersections of regular languages are again regular.
  On the other hand,
  complements of regular languages are  not always regular.
  We introduce a width-bounded complement
  and show that width-bounded complements of regular languages are again regular.

  We also study determinism and ambiguity. 
  We show that it is decidable whether a regular language is accepted by a deterministic HDA and that there exists regular languages with unbounded ambiguity.
  Finally, we characterize one-letter deterministic languages in terms of utlimately periodic functions.
\end{abstract}

\end{frontmatter}

\section{Introduction}
\label{s:intro}

\begin{figure}[bp]
  \centering
  \begin{tikzpicture}
    \begin{scope}[x=1.5cm, state/.style={shape=circle, draw,
        fill=white, initial text=, inner sep=1mm, minimum size=3mm}]
      \node[state, black] (10) at (0,0) {};
      \node[state, rectangle] (20) at (0,-1) {$\vphantom{b}a$};
      \node[state] (30) at (0,-2) {};
      \node[state, black] (11) at (1,0) {};
      \node[state, rectangle] (21) at (1,-1) {$b$};
      \node[state] (31) at (1,-2) {};
      \path (10) edge (20);
      \path (20) edge (30);
      \path (11) edge (21);
      \path (21) edge (31);
      \node[state, black] (m) at (.5,-1) {};
      \path (20) edge[out=15, in=165] (m);
      \path (m) edge[out=-165, in=-15] (20);
      \path (21) edge[out=165, in=15] (m);
      \path (m) edge[out=-15, in=-165] (21);
    \end{scope}
    \begin{scope}[xshift=4cm]
      \node[state] (00) at (0,0) {};
      \node[state] (10) at (-1,-1) {};
      \node[state] (01) at (1,-1) {};
      \node[state] (11) at (0,-2) {};
      \path (00) edge node[left] {$\vphantom{b}a$\,} (10);
      \path (00) edge node[right] {\,$b$} (01);
      \path (10) edge node[left] {$b$\,} (11);
      \path (01) edge node[right] {\,$\vphantom{b}a$} (11);
    \end{scope}
  \end{tikzpicture}
  \qquad\qquad
  \begin{tikzpicture}
    \begin{scope}[xshift=8cm]
      \path[fill=black!15] (0,0) to (-1,-1) to (0,-2) to (1,-1);
      \node[state] (00) at (0,0) {};
      \node[state] (10) at (-1,-1) {};
      \node[state] (01) at (1,-1) {};
      \node[state] (11) at (0,-2) {};
      \path (00) edge node[left] {$\vphantom{b}a$\,} (10);
      \path (00) edge node[right] {\,$b$} (01);
      \path (10) edge node[left] {$b$\,} (11);
      \path (01) edge node[right] {\,$\vphantom{b}a$} (11);
    \end{scope}
    \begin{scope}[x=1.5cm, state/.style={shape=circle, draw,
        fill=white, initial text=, inner sep=1mm, minimum size=3mm},
      xshift=10.5cm]
      \node[state, black] (10) at (0,0) {};
      \node[state, rectangle] (20) at (0,-1) {$\vphantom{b}a$};
      \node[state] (30) at (0,-2) {};
      \node[state, black] (11) at (1,0) {};
      \node[state, rectangle] (21) at (1,-1) {$b$};
      \node[state] (31) at (1,-2) {};
      \path (10) edge (20);
      \path (20) edge (30);
      \path (11) edge (21);
      \path (21) edge (31);
    \end{scope}
  \end{tikzpicture}
  \caption{Petri net and HDA models
    distinguishing interleaving (left)
    from non-interleaving (right) concurrency.
    Left: models for~$a. b+ b. a$;
    right: models for~$a\para b$.}
  \label{fi:int-conc}
\end{figure}

Higher-dimensional automata (HDAs),
introduced by Pratt and van Glabbeek \cite{Pratt91-geometry, Glabbeek91-hda},
are a general geometric model for non-interleaving concurrency
which subsumes, for example, event structures and
Petri nets~\cite{DBLP:journals/tcs/Glabbeek06}.
HDAs of dimension one are standard automata,
whereas HDAs of dimension two are isomorphic to asynchronous transition systems~\cite{%
  Bednarczyk87-async, DBLP:journals/cj/Shields85, Goubault02-cmcim}.
As an example, Figure~\ref{fi:int-conc} shows Petri net and HDA models
for a system with two events, labeled $a$ and $b$.
The Petri net and HDA on the left side
model the (mutually exclusive) interleaving of $a$ and $b$ as either $a. b$ or $b. a$;
those to the right model concurrent execution of $a$ and $b$.
In the HDA, this independence is indicated by a filled-in square.

Recent work defines languages of HDAs~\cite{DBLP:journals/mscs/FahrenbergJSZ21},
which are sets of partially ordered multisets with interfaces (ipomsets)~\cite{DBLP:journals/iandc/FahrenbergJSZ22}
that are closed under subsumptions.
Follow-up papers introduce a language theory for HDAs,
showing a Kleene theorem~\cite{DBLP:conf/concur/FahrenbergJSZ22},
which makes a connection between rational and regular ipomset languages
(those accepted by finite HDAs),
and a Myhill-Nerode theorem~\cite{DBLP:journals/corr/abs-2210-08298}
stating that regular languages are precisely those that have finite prefix quotient.
Here we continue to develop this nascent higher-dimensional automata theory.

Our first contribution, in Section~\ref{se:stuff}, is a pumping lemma for HDAs,
based on the fact that if an ipomset accepted by an HDA is long enough,
then there is a cycle in the path that accepts it.
As an application we can expose a class of non-regular ipomset languages.
We also show that regular languages are closed under intersection,
both using the Myhill-Nerode theorem and an explicit product construction.

In Section~\ref{se:st} we introduce a translation from HDAs
to ordinary finite automata over an alphabet of discrete ipomsets,
called ST-automata.
The translation forgets some of the structure of the HDA;
nevertheless, it allows us to show that
inclusion of regular ipomset languages is decidable.
This immediately implies that emptiness is decidable;
universality is trivial given that the universal language is not regular.

The paper~\cite{DBLP:journals/corr/abs-2210-08298}
introduces deterministic HDAs and shows that not all HDAs are determinizable.
We show that the question whether a language is deterministic is decidable.
As a weaker notion in-between determinism and non-determinism,
one may ask whether all regular languages may be recognized by finitely ambiguous HDAs,
\ie~those in which there is an upper bound for the number of accepting paths on any ipomset.
We show in Section~\ref{se:det} that the answer to this question is negative and that there are regular languages of unbounded ambiguity.

In Section~\ref{se:one}, we discuss one-letter languages. For regular languages,
\cite{journals/mathgrund/Buchi60} develops a characterization based on ultimately periodic functions. We extend this to regular languages on pomsets with one letter. However, we show that also one-letter languages may be non-deterministic and of unbounded ambiguity.

Finally, in Section~\ref{se:comp}, we are interested in a notion of complement.
This immediately raises two problems:
first, complements of ipomset languages are generally not closed under subsumption;
second, the complement of the empty language, which is regular,
is the universal language, which is non-regular.
The first problem is solved by taking subsumption closure,
turning complement into a pseudocomplement in the sense of lattice theory.

As to the second problem, we can show that complements of regular languages are non-regular.
Yet if we restrict the width of our languages, \ie the number of events which may occur concurrently,
then the so-defined width-bounded complement has good properties:
it is still a pseudocomplement;
the languages for which double complement is identity
have an easy characterization;
and finally width-bounded complements of regular languages are again regular.
The proof of that last property again uses ST-automata
and the fact that the induced translation from ipomset languages to word languages over discrete ipomsets
has good algebraic properties.
We note that width-bounded languages and (pseudo)complements are found in other works on concurrent languages,
for example \cite{DBLP:journals/jlp/JipsenM16, DBLP:journals/tcs/LodayaW00, DBLP:conf/concur/FanchonM02}.

Another goal of this work was to obtain the above results using automata-theoretic means
as opposed to category-theoretic or topological ones.
Indeed we do not use presheaves, track objects, cylinders,
or any other of the categorical or topological constructions
employed in~\cite{DBLP:conf/concur/FahrenbergJSZ22, DBLP:journals/corr/abs-2210-08298}.
Categorical reasoning would have simplified proofs in several places,
but no background in category theory or algebraic topology is necessary to understand this paper.

To sum up, our main contributions to higher-dimensional automata theory are as follows:
\begin{itemize}
\item a pumping lemma (Lemma~\ref{le:localpump});
\item closure of regular languages under intersection (Proposition~\ref{pr:cap});
\item decidability of inclusion of regular languages (Theorem~\ref{th:incl});
\item decidability of determinism of a regular language (Theorem~\ref{th:deterministic});
\item regular languages of unbounded ambiguity (Proposition~\ref{pr:lambig});
\item a characterization of one-letter deterministic languages (Proposition~\ref{p:DetOneLetUlPer});
\item closure of regular languages under width-bounded complement (Theorem~\ref{th:compClosure}).
\end{itemize}
This paper is an extended version of \cite{DBLP:conf/ictac/AmraneBFZ23}
which was presented at the 20th International Colloquium on Theoretical Aspects of Computing (ICTAC 2023).
Compared to this previous work, we update the definition of ST-automata to the one introduced in \cite{conf/ramics/AmraneBCFZ24},
we describe one-letter languages,
and we provide an effective procedure to decide whether a regular language is deterministic.
The latter solves an open problem posed in \cite{DBLP:conf/ictac/AmraneBFZ23}.
Also in~\cite{DBLP:conf/ictac/AmraneBFZ23}, we sought further development of the theory of timed HDAs and a Büchi-like characterization of HDA languages. This has been achieved in \cite{DBLP:conf/apn/AmraneBCF24} and \cite{conf/dlt/AmraneBFF24}, respectively.

\section{Pomsets with interfaces}

HDAs model systems in which (labeled)
events have duration and may happen concurrently.
Notably, as seen in the introduction,
concurrency of events is a more general notion than interleaving.
Every event has an interval in time during which it is active:
it starts at some point in time, then remains active until it terminates,
and never appears again.
Events may be concurrent,
in which case their activity intervals overlap:
one of the two events starts before the other terminates.
Executions are thus isomorphism classes of partially ordered intervals.
For reasons of compositionality
we also consider executions in which events may be active already at the beginning
or remain active at the end.

Any time point of an execution defines a \emph{concurrency list} (or \emph{conclist})
of currently active events.
The relative position of any two concurrent events on such lists
does not change during passage of time;
this equips events of an execution
with a partial order which we call \emph{event order}.
The temporal order of non-concurrent events
(one of two events terminating before the other starts)
introduces another partial order which we call \emph{precedence}.
An execution is, then, a collection of labeled events together with two partial orders.

To make the above precise, let $\Sigma$ be a finite alphabet.
We define three notions, in increasing order of generality:
\begin{itemize}
\item A \emph{concurrency list}, or \emph{conclist},
  $U=(U, {\evord_U}, \lambda_U)$
  consists of a finite set $U$,
  a strict total order ${\evord_U}\subseteq U\times U$ (the event order),\footnote{%
  A strict \emph{partial} order is a relation which is irreflexive and transitive;
  a strict \emph{total} order is a relation which is irreflexive, transitive, and total.
  We may omit the ``strict''.}
  and a labeling $\lambda_U: U\to \Sigma$.
\item A \emph{partially ordered multiset}, or \emph{pomset},
  $P=(P, {<_P}, {\evord_P}, \lambda_P)$
  consists of a finite set $P$,
  two strict partial orders ${<_P}, {\evord_P}\subseteq P\times P$
  (precedence and event order),
  and a labeling $\lambda_P: P\to \Sigma$,
  such that for each $x\ne y$ in $P$,
  at least one of
  $x<_P y$, $y<_P x$, $x\evord_P y$, or $y\evord_P x$ holds.
\item A \emph{pomset with interfaces}, or \emph{ipomset},
  $(P, {<_P}, {\evord_P}, S_P, T_P, \lambda_P)$
  consists of a pomset $(P, {<_P}, {\evord_P}, \lambda_P)$
  together with subsets $S_P, T_P\subseteq P$
  (\emph{source} and \emph{target} \emph{interfaces})
  such that elements of $S_P$ are \mbox{$<_P$-minimal} and those of $T_P$ are \mbox{$<_P$-maximal}.
\end{itemize}
We will omit the subscripts $_U$ and $_P$ whenever possible.
We will also often use the notation $\ilo{S}{P}{T}$ instead of $(P, {<}, {\evord}, S, T, \lambda)$ if no confusion may arise.

\begin{figure}[tbp]
  \centering
  \begin{tikzpicture}[x=.95cm, scale=1, every node/.style={transform shape}]
    \def\possh{-1.3}
    \begin{scope}[shift={(9.6,0)}]
      \def\hw{0.3}
      \filldraw[fill=green!50!white,-](0,1.2)--(1.2,1.2)--(1.2,1.2+\hw)--(0,1.2+\hw);
      \filldraw[fill=pink!50!white,-](0.3,0.7)--(1.9,0.7)--(1.9,0.7+\hw)--(0.3,0.7+\hw)--(0.3,0.7);
      \filldraw[fill=blue!20!white,-](0.5,0.2)--(1.7,0.2)--(1.7,0.2+\hw)--(0.5,0.2+\hw)--(0.5,0.2);
      \draw[thick,-](0,0)--(0,1.7);
      \draw[thick,-](2.2,0)--(2.2,1.7);
      \node at (0.6,1.2+\hw*0.5) {$a$};
      \node at (1.1,0.7+\hw*0.5) {$b$};
      \node at (1.1,0.2+\hw*0.5) {$c$};
    \end{scope}
    \begin{scope}[shift={(9.6,\possh)}]
      \node (a) at (0.4,0.7) {$\ibullet a$};
      \node (c) at (0.4,-0.7) {$c$};
      \node (b) at (1.8,0) {$b$};
      \path[densely dashed, gray] (a) edge (b) (b) edge (c) (a) edge (c);
    \end{scope}
    \begin{scope}[shift={(6.4,0)}]
      \def\hw{0.3}
      \filldraw[fill=green!50!white,-](0,1.2)--(1.2,1.2)--(1.2,1.2+\hw)--(0,1.2+\hw);
      \filldraw[fill=pink!50!white,-](1.3,0.7)--(1.9,0.7)--(1.9,0.7+\hw)--(1.3,0.7+\hw)--(1.3,0.7);
      \filldraw[fill=blue!20!white,-](0.5,0.2)--(1.7,0.2)--(1.7,0.2+\hw)--(0.5,0.2+\hw)--(0.5,0.2);
      \draw[thick,-](0,0)--(0,1.7);
      \draw[thick,-](2.2,0)--(2.2,1.7);
      \node at (0.6,1.2+\hw*0.5) {$a$};
      \node at (1.6,0.7+\hw*0.5) {$b$};
      \node at (1.1,0.2+\hw*0.5) {$c$};
    \end{scope}
    \begin{scope}[shift={(6.4,\possh)}]
      \node (a) at (0.4,0.7) {$\ibullet a$};
      \node (c) at (0.4,-0.7) {$c$};
      \node (b) at (1.8,0) {$b$};
      \path (a) edge (b);
      \path[densely dashed, gray]  (b) edge (c) (a) edge (c);
    \end{scope}
    \begin{scope}[shift={(3.2,0)}]
      \def\hw{0.3}
      \filldraw[fill=green!50!white,-](0,1.2)--(1.2,1.2)--(1.2,1.2+\hw)--(0,1.2+\hw);
      \filldraw[fill=pink!50!white,-](1.3,0.7)--(1.9,0.7)--(1.9,0.7+\hw)--(1.3,0.7+\hw)--(1.3,0.7);
      \filldraw[fill=blue!20!white,-](0.5,0.2)--(1.1,0.2)--(1.1,0.2+\hw)--(0.5,0.2+\hw)--(0.5,0.2);
      \draw[thick,-](0,0)--(0,1.7);
      \draw[thick,-](2.2,0)--(2.2,1.7);
      \node at (0.6,1.2+\hw*0.5) {$a$};
      \node at (1.6,0.7+\hw*0.5) {$b$};
      \node at (0.8,0.2+\hw*0.5) {$c$};
    \end{scope}
    \begin{scope}[shift={(3.2,\possh)}]
      \node (a) at (0.4,0.7) {$\ibullet a$};
      \node (c) at (0.4,-0.7) {$c$};
      \node (b) at (1.8,0) {$b$};
      \path (a) edge (b) (c) edge (b);
      \path[densely dashed, gray]  (a) edge (c);
    \end{scope}
    \begin{scope}[shift={(0.0,0)}]
      \def\hw{0.3}
      \filldraw[fill=green!50!white,-](0,1.2)--(0.4,1.2)--(0.4,1.2+\hw)--(0,1.2+\hw);
      \filldraw[fill=pink!50!white,-](1.3,0.7)--(1.9,0.7)--(1.9,0.7+\hw)--(1.3,0.7+\hw)--(1.3,0.7);
      \filldraw[fill=blue!20!white,-](0.5,0.2)--(1.1,0.2)--(1.1,0.2+\hw)--(0.5,0.2+\hw)--(0.5,0.2);
      \draw[thick,-](0,0)--(0,1.7);
      \draw[thick,-](2.2,0)--(2.2,1.7);
      \node at (0.2,1.2+\hw*0.5) {$a$};
      \node at (1.6,0.7+\hw*0.5) {$b$};
      \node at (0.8,0.2+\hw*0.5) {$c$};
    \end{scope}
    \begin{scope}[shift={(0.0,\possh)}]
      \node (a) at (0.4,0.7) {$\ibullet a$};
      \node (c) at (0.4,-0.7) {$c$};
      \node (b) at (1.8,0) {$b$};
      \path (a) edge (b) (c) edge (b) (a) edge (c);
    \end{scope}
  \end{tikzpicture}
  \caption{Activity intervals of events (top)
    and corresponding ipomsets (bottom),
    \cf~Example~\ref{ex:subsu}.
    Full arrows indicate precedence order;
    dashed arrows indicate event order;
    bullets indicate interfaces.}
  \label{fi:iposets1}
\end{figure}

Conclists may be regarded as pomsets with empty precedence
(\emph{discrete} pomsets);
the last condition
above enforces that $\evord$ is then total.
Pomsets are ipomsets with empty interfaces,
and in any ipomset $P$, the substructures induced by $S_P$ and $T_P$ are conclists.
Note that different events of ipomsets may carry the same label;
in particular we do \emph{not} exclude autoconcurrency.
Figure~\ref{fi:iposets1} shows some simple examples.
Source and target events are marked by ``$\ibullet$'' at the left or right side,
and if the event order is not shown, we assume that it goes downwards.

An ipomset $P$ is \emph{interval}
if $<_P$ is an interval order~\cite{book/Fishburn85},
that is, if it admits an interval representation
given by functions $b$ and $e$ from $P$ to real numbers such that
$b(x)\le e(x)$ for all $x\in P$ and
$x<_P y$ iff $e(x)<b(y)$ for all $x, y\in P$.
Given that our ipomsets represent activity intervals of events,
any of the ipomsets we will encounter will be interval,
and we omit the qualification ``interval''.
We emphasise that this is \emph{not} a restriction, but rather induced by the semantics,
see also~\cite{Wiener14}.
We let $\iiPoms$ denote the set of (interval) ipomsets.

Ipomsets may be \emph{refined} by shortening
activity intervals,
potentially removing concurrency and expanding precedence.
The inverse to refinement is called \emph{subsumption} and defined as follows.
For ipomsets $P$ and $Q$ we say that $Q$ subsumes $P$ and write $P\subsu Q$
if there is a bijection $f: P\to Q$ for which
\begin{enumerate}[(1)]
\item $f(S_P)=S_Q$, $f(T_P)=T_Q$, and $\lambda_Q\circ f=\lambda_P$;
\item \label{en:subsu.prec}
  $f(x)<_Q f(y)$ implies $x<_P y$;
\item \label{en:subsu.evord}
  $x\not<_P y$, $y\not<_P x$ and $x\evord_P y$ imply $f(x)\evord_Q f(y)$.
\end{enumerate}
That is, $f$ respects interfaces and labels, reflects precedence, and preserves essential event order.
(Event order is essential for concurrent events,
but by transitivity, it also appears between non-concurrent events.
Subsumptions ignore such non-essential event order.)
This definition adapts the one of~\cite{DBLP:journals/fuin/Grabowski81} to event orders and interfaces.
Intuitively, $P$ has more order and less concurrency than $Q$.

\begin{example}
  \label{ex:subsu}
  In Figure~\ref{fi:iposets1} there is a sequence of subsumptions from left to right:
  \begin{equation*}
  \ibullet acb \subsu
  \pomset[1]{\ibullet \scriptstyle a \scriptstyle \ar[dr] \\ & \scriptstyle  b \\ \; \scriptstyle c \scriptstyle \ar[ur]} \subsu
  \loset{\ibullet a \to b \\ \hphantom{\ibullet}c} \subsu
  \loset{\ibullet a \\ \hphantom{\ibullet}b \\ \hphantom{\ibullet}c}
  \end{equation*}
  An event $e_1$ is smaller than $e_2$ in the precedence order if $e_1$ is terminated before $e_2$ is started;
  $e_1$ is smaller than $e_2$ in the event order if they are concurrent and $e_1$ is above $e_2$ in the respective conclist.
\end{example}

\emph{Isomorphisms} of ipomsets are invertible subsumptions,
\ie bijections $f$ for which items \ref{en:subsu.prec} and \ref{en:subsu.evord} above
are strengthened to
\begin{enumerate}[(1$'$)]
  \setcounter{enumi}1
\item $f(x)<_Q f(y)$ iff $x<_P y$;
\item $x\not<_P y$ and $y\not<_P x$ imply that $x\evord_P y$ iff $f(x)\evord_Q f(y)$.
\end{enumerate}
Due to the requirement that all elements are ordered by $<$ or $\evord$,
there is at most one isomorphism between any two ipomsets.
Hence we may switch freely between ipomsets and their isomorphism classes.
We will also call these equivalence classes ipomsets and often conflate equality and isomorphism.

\subsection{Compositions}

The standard serial and parallel compositions of pomsets~\cite{DBLP:journals/fuin/Grabowski81}
extend to ipomsets.
The \emph{parallel} composition of ipomsets $P$ and $Q$ is
$P\parallel Q=(P\sqcup Q, {<}, {\evord}, S, T, \lambda)$, where
$P\sqcup Q$ denotes disjoint union and
\begin{itemize}
\item $x<y$ if $x<_P y$ or $x<_Q y$;
\item $x\evord y$ if $x\evord_P y$, $x\evord_Q y$, or $x\in P$ and $y\in Q$;
\item $S=S_P\cup S_Q$ and $T=T_P\cup T_Q$;
\item $\lambda(x)=\lambda_P(x)$ if $x\in P$ and $\lambda(x)=\lambda_Q(x)$ if $x\in Q$.
\end{itemize}
Note that parallel composition of ipomsets is generally not commutative,
see \cite{DBLP:journals/iandc/FahrenbergJSZ22} or Example~\ref{ex:compnotcup} below for details.

\begin{figure}[tbp]
  \centering
  \begin{tikzpicture}[x=.35cm, y=.5cm]
    \begin{scope}
      \node (1) at (2,2) {$\vphantom{bd}a$};
      \node (2) at (0,0) {$b$};
      \node (3) at (4,0) {$\vphantom{bd}c\ibullet$};
      \path (2) edge (3);
      \path (1) edge[densely dashed, gray] (2);
      \path (3) edge[densely dashed, gray] (1);
      \node (ast) at (5.5,1) {$\ast$};
      \node (4) at (7,2) {$d$};
      \node (5) at (7,0) {$\vphantom{bd}\ibullet c$};
      \path (4) edge[densely dashed, gray] (5);
      \node (equals) at (8.5,1) {$=$};
      \node (6) at (10,2) {$\vphantom{bd}a$};
      \node (7) at (10,0) {$b$};
      \node (8) at (14,2) {$d$};
      \node (9) at (14,0) {$\vphantom{bd}c$};
      \path (6) edge (8);
      \path (7) edge (9);
      \path (7) edge (8);
      \path (8) edge[densely dashed, gray] (9);
      \path (9) edge[densely dashed, gray] (6);
      \path (6) edge[densely dashed, gray] (7);
    \end{scope}
    \begin{scope}[shift={(19,0)}]
      \path[use as bounding box] (0,-2.5) -- (14,3.2);
      \node (1) at (2,2) {$\vphantom{bd}a$};
      \node (2) at (0,0) {$b$};
      \node (3) at (4,0) {$\vphantom{bd}c\ibullet$};
      \path (2) edge (3);
      \path (1) edge[densely dashed, gray] (2);
      \path (3) edge[densely dashed, gray] (1);
      \node (ast) at (5.5,1) {$\parallel$};
      \node (4) at (7,2) {$d$};
      \node (5) at (7,0) {$\vphantom{bd}\ibullet c$};
      \path (4) edge[densely dashed, gray] (5);
      \node (equals) at (8.5,1) {$=$};
      \begin{scope}[shift={(0,-.9)}]
        \node (6) at (12,4) {$\vphantom{bd}a$};
        \node (7) at (10,2) {$b$};
        \node (8) at (14,2) {$\vphantom{bd}c\ibullet$};
        \path (7) edge (8);
        \path (6) edge[densely dashed, gray] (7);
        \path (8) edge[densely dashed, gray] (6);
        \node (9) at (12,0) {$d$};
        \node (10) at (12,-2) {$\vphantom{bd}\ibullet c$};
        \path (9) edge[densely dashed, gray] (10);
        \path(7) edge[densely dashed, gray] (9);
      \end{scope}
    \end{scope}
  \end{tikzpicture}
  \caption{Gluing and parallel composition of ipomsets.}
  \label{fi:compositions}
\end{figure}

Serial composition generalises to a \emph{gluing} composition
which continues interface events across compositions and is defined as follows.
Let $P$ and $Q$ be ipomsets such that $T_P=S_Q$,
$x\evord_P y$ iff $x\evord_Q y$ for all $x, y\in T_P=S_Q$, and
the restrictions $\lambda_P\rest{T_P}=\lambda_Q\rest{S_Q}$,
then $P*Q=(P\cup Q, {<}, {\evord}, S_P, T_Q, \lambda)$, where
\begin{itemize}
\item $x<y$ if $x<_P y$, $x<_Q y$, or $x\in P-T_P$ and $y\in Q-S_Q$;\footnote{%
    We use ``$-$'' for set difference instead of the perhaps more common ``$\setminus$''.}
\item $\evord$ is the transitive closure of ${\evord_P}\cup {\evord_Q}$;
\item $\lambda(x)=\lambda_P(x)$ if $x\in P$ and $\lambda(x)=\lambda_Q(x)$ if $x\in Q$.
\end{itemize}
Gluing is, thus, only defined if the targets of $P$ are equal to the sources of $Q$ 
\emph{as conclists}.
If we would not conflate equality and isomorphism,
we would have to define the carrier set of $P*Q$
to be the disjoint union of $P$ and $Q$ quotiented out by the unique isomorphism $T_P\to S_Q$.
We will often omit the ``$*$'' in gluing compositions.
Fig~\ref{fi:compositions} shows some examples.

An ipomset $P$ is a~\emph{word} (with interfaces) if $<_P$ is total.
Conversely, $P$ is \emph{discrete} if $<_P$ is empty (hence $\evord_P$ is total).
Conclists are discrete ipomsets without interfaces.
The relation $\subsu$ is a partial order on $\iiPoms$
with minimal elements words
and maximal elements discrete ipomsets.
Further, gluing and parallel compositions respect $\subsu$.

\subsection{Special ipomsets}

A \emph{starter} is a discrete ipomset $U$ with $T_U=U$,
a \emph{terminator} one with $S_U=U$.
The intuition is that a starter does nothing but start the events in $A=U-S_U$,
and a terminator terminates the events in $B=U-T_U$.
These will be so important later that we introduce special notation,
writing $\starter{U}{A} = \ilo{U\setminus A}{U}{U}$ and $\terminator{U}{B}=\ilo{U}{U}{U\setminus B}$ for the above.
Starter $\starter{U}{A}$ is \emph{elementary} if $A$ is a singleton,
similarly for $\terminator{U}{B}$.
Discrete ipomsets $U$ with $S_U=T_U=U$ are identities for the gluing composition and written~$\id_U$.
Note that $\id_U=\starter{U}{\emptyset}=\terminator{U}{\emptyset} = \ilo{U}{U}{U}$.

The \emph{width} $\wid(P)$ of an ipomset $P$ is the cardinality of a maximal $<$-antichain.
For $k \ge 0$, we let $\iiPoms_{\le k} \subseteq \iiPoms$ denote the set of ipomsets of width at most $k$.
The \emph{size} of an ipomset $P$ is $\size(P)=|P|-\tfrac{1}{2}(|S_P|+|T_P|)$.
Identities are exactly the ipomsets of size $0$.
Elementary starters and terminators are exactly the ipomsets of size~$\tfrac 1 2$.

\begin{figure}[tbp]
  \centering
  \begin{subfigure}{0.2\linewidth}
    \begin{tikzpicture}
      \node (a) at (0.4,1.4) {$\ibullet a$};
      \node (c) at (0.4,0.7) {$c$};
      \node (b) at (1.8,1.4) {$b$};
      \node (d) at (0.4,0) {$d$};
      \path (a) edge (b);
      \path[densely dashed, gray]  (c) edge (d) (b) edge (c) (a) edge (c);
    \end{tikzpicture}
  \end{subfigure}
  \begin{subfigure}{0.65\linewidth}
    Sparse: $\normalsize \loset{\ibullet a \ibullet \\ \phantom{\ibullet} c \ibullet \\ \phantom{\ibullet} d \ibullet} * \loset{\ibullet a \phantom{\ibullet}\\ \ibullet c \ibullet \\ \ibullet d \ibullet} * \loset{\phantom{\ibullet}b \ibullet \\ \ibullet c \ibullet \\ \ibullet d \ibullet} * \loset{\ibullet b \\ \ibullet c \\ \ibullet d}$

    \medskip
    Dense: $\normalsize \loset{\ibullet a \ibullet \\ \phantom{\ibullet}c \ibullet} * \loset{\ibullet a \ibullet \\ \ibullet c \ibullet \\ \phantom{\ibullet}d \ibullet} * \loset{\ibullet a\phantom{\ibullet} \\ \ibullet c \ibullet \\ \ibullet d \ibullet} * \loset{\phantom{\ibullet}b \ibullet \\ \ibullet c \ibullet \\ \ibullet d \ibullet} * \loset{\ibullet b\phantom{\ibullet} \\ \ibullet c \ibullet \\ \ibullet d \ibullet} * \loset{\ibullet c \ibullet \\ \ibullet d\phantom{\ibullet}} * \ibullet c$
  \end{subfigure}
  \caption{Ipomset of size $3.5$ and two of its step decompositions.}
  \label{fi:densesparse}
\end{figure}

Any ipomset can be decomposed as a gluing of starters and terminators~\cite{DBLP:journals/iandc/FahrenbergJSZ22},
see also~\cite{DBLP:journals/fuin/JanickiK19}.
Such a presentation we call a \emph{step decomposition}.
If starters and terminators are alternating, the step decomposition is called \emph{sparse};
if they are all elementary, then it is \emph{dense}. The algebra of these representations is studied more in depth in~\cite{conf/ramics/AmraneBCFZ24}, along with how to represent subsumptions in this formalism.

\begin{example}
  Figure~\ref{fi:densesparse} illustrates two step decompositions.
  The sparse one first starts $c$ and $d$, then terminates $a$, starts $b$, and terminates $b$, $c$ and $d$ together.
  The dense one first starts $c$, then starts $d$, terminates $a$, starts $b$,
  and finally terminates $b$, $d$, and $c$ in order.
\end{example}

\begin{lemma}[\cite{DBLP:journals/corr/abs-2210-08298}]
  \label{le:ipomsparse}
  Every ipomset $P$ has a unique sparse step decomposition.
\end{lemma}

Dense step decompositions are generally not unique, but they all have the same length.

\begin{lemma}
  \label{le:ipomdense}
  Every dense step decomposition of ipomset $P$ has length $2\,\size(P)$.
\end{lemma}

\begin{proof}
  Every element of a dense step decomposition of $P$
  starts precisely one event
  or terminates precisely one event.
  Thus every event in $P-(S_P\cup T_P)$ gives rise to two elements in the step decomposition
  and every event in $S_P\cup T_P-(S_P\cap T_P)$ to one element.
  The length of the step decomposition is, thus,
  $2 |P| - 2 |S_P\cup T_P| + |S_P\cup T_P| - |S_P\cap T_P|
  = 2 |P| - (|S_P|+|T_P|-|S_P\cap T_P|) - |S_P\cap T_P| = 2\, \size(P)$.
\end{proof}

\subsection{Rational languages}

For $A\subseteq \iiPoms$ we let
\begin{equation*}
  A\down=\{P\in \iiPoms\mid \exists Q\in A: P\subsu Q\}.
\end{equation*}
Note that $(A\cup B)\down=A\down\cup B\down$ for all $A, B\subseteq \iiPoms$,
but for intersection this does \emph{not} hold.
For example it may happen that $A\cap B=\emptyset$ but $A\down\cap B\down\ne \emptyset$.
A \emph{language} is a subset $L\subseteq \iiPoms$ for which $L\down=L$.
The set of all languages is denoted $\Langs\subseteq 2^{\iiPoms}$.

The \emph{width} of a language $L$ is $\wid(L)=\sup\{\wid(P)\mid P\in L\}$.
For $k\ge 0$ and $L\in \Langs$, denote $L_{\le k}=\{P\in L\mid \wid(P)\le k\}$.
$L$ is \emph{$k$-dimensional} if $L=L_{\le k}$.
We let $\Langs_{\le k}=\Langs\cap \iiPoms_{\le k}$ denote the set of $k$-dimensional languages.

The \emph{singleton ipomsets} are
$[a]$ $[\ibullet a]$, $[a\ibullet]$ and $[\ibullet a\ibullet]$,
for all $a\in \Sigma$.
The \emph{rational operations} $\cup$, $*$, $\|$ and (Kleene plus) $^+$ for languages are defined as follows.
\begin{align*}
  L*M &= \{P*Q\mid P\in L,\; Q\in M,\; T_P=S_Q\}\down, \\
  L\parallel M &= \{P\parallel Q\mid P\in L,\; Q\in M\}\down, \\
  L^+ &= \smash[t]{\bigcup\nolimits_{n\ge 1}} L^n,\qquad \text{ for } L^1=L, L^{n+1}= L*L^n.
\end{align*}
The class of \emph{rational languages} is the smallest subset of $\Langs$ that contains
\begin{equation*}
  \big\{\emptyset, \{\epsilon\}, \{[a]\}, \{[\ibullet a]\},\{[a\ibullet]\}, \{[\ibullet a\ibullet]\}\mid a\in \Sigma\big\}
\end{equation*}
($\epsilon$ denotes the empty ipomset)
and is closed under the rational operations. 

\begin{lemma}[\cite{DBLP:conf/concur/FahrenbergJSZ22}]
  \label{lem:finitewidth}
  Any rational language has finite width.
\end{lemma}

It immediately follows that the universal language $\iiPoms$ is \emph{not} rational.

The \emph{prefix quotient} of a language $L\in \Langs$ by an ipomset $P$ is $P\backslash L=\{Q\in \iiPoms\mid P Q\in L\}$.
Similarly, the \emph{suffix quotient} of $L$ by $P$ is $L/P=\{Q\in \iiPoms\mid Q P\in L\}$.
Denoting
\begin{equation*}
  \suff(L)=\{P\backslash L\mid P\in \iiPoms\},
  \qquad
  \pref(L)=\{L/P\mid P\in\iiPoms\},
\end{equation*}
we may now state the central result of \cite{DBLP:journals/corr/abs-2210-08298}.

\begin{theorem}[\cite{DBLP:journals/corr/abs-2210-08298}]
  \label{th:MN}
  Let $L\in\Langs$ be a language.  The following are equivalent:
  \begin{enumerate}
  \item $L$ is rational;
  \item $\suff(L)$ is finite;
  \item $\pref(L)$ is finite.
  \end{enumerate}
\end{theorem}

\section{Higher-dimensional automata}

An HDA is a collection of \emph{cells} which are connected by \emph{face maps}.
Each cell contains a conclist of events which are active in it,
and the face maps may terminate some events (\emph{upper} faces)
or ``unstart'' some events (\emph{lower} faces),
\ie map a cell to another in which the indicated events are not yet active.

To make this precise,
let $\square$ denote the set of conclists.
A \emph{precubical set}
\begin{equation*}
  X=(X, {\ev}, \{\delta_{A, U}^0, \delta_{A, U}^1\mid U\in \square, A\subseteq U\})
\end{equation*}
consists of a set of cells $X$
together with a function $\ev: X\to \square$.
For a conclist $U$ we write $X[U]=\{x\in X\mid \ev(x)=U\}$ for the cells of type $U$.
Further, for every $U\in \square$ and $A\subseteq U$ there are face maps
$\delta_{A}^0, \delta_{A}^1: X[U]\to X[U-A]$
which satisfy $\delta_A^\nu \delta_B^\mu = \delta_B^\mu \delta_A^\nu$
for $A\cap B=\emptyset$ and $\nu, \mu\in\{0, 1\}$.
The upper face maps $\delta_A^1$ transform a cell $x$ into one in which the events in $A$ have terminated,
whereas the lower face maps $\delta_A^0$ transform $x$ into a cell where the events in $A$ have not yet started.
The \emph{precubical identity} above
expresses the fact that these transformations commute for disjoint sets of events.

\begin{figure}[tbp]
  \centering
  \begin{tikzpicture}[x=.9cm, y=.8cm, scale=0.9, every node/.style={transform shape}]
    \node[circle,draw=black,fill=blue!30,inner sep=0pt,minimum size=15pt]
    (aa) at (0,0) {$\vphantom{hy}v$};
    \node[circle,draw=black,fill=blue!30,inner sep=0pt,minimum size=15pt]
    (ac) at (0,4) {$\vphantom{hy}x$};
    \node[circle,draw=black,fill=blue!30,inner sep=0pt,minimum size=15pt]
    (ca) at (4,0) {$\vphantom{hy}w$};
    \node[circle,draw=black,fill=blue!30,inner sep=0pt,minimum size=15pt]
    (cc) at (4,4) {$\vphantom{hy}y$};
    \node[circle,draw=black,fill=red!30,inner sep=0pt,minimum size=15pt]
    (ba) at (2,0) {$\vphantom{hy}e$};
    \node[circle,draw=black,fill=red!30,inner sep=0pt,minimum size=15pt]
    (bc) at (2,4) {$\vphantom{hy}f$};
    \node[circle,draw=black,fill=green!30,inner sep=0pt,minimum size=15pt]
    (ab) at (0,2) {$\vphantom{hy}g$};
    \node[circle,draw=black,fill=green!30,inner sep=0pt,minimum size=15pt]
    (cb) at (4,2) {$\vphantom{hy}h$};
    \node[circle,draw=black,fill=black!20,inner sep=0pt,minimum size=15pt]
    (bb) at (2,2) {$\vphantom{hy}q$};
    \node[right] at (5,4) {$X[\emptyset]=\{v,w,x,y\}$};
    \node[right] at (5,3.2) {$X[a]=\{e,f\}$};
    \node[right] at (5,2.4) {$X[b]=\{g,h\}$};
    \node[right] at (5,1.6) {$X[\loset{a\\b}]=\{q\}$};
    \path (ba) edge node[above] {$\delta^0_a$} (aa);
    \path (ba) edge node[above] {$\delta^1_a$} (ca);
    \path (bb) edge node[above] {$\delta^0_a$} (ab);
    \path (bb) edge node[above] {$\delta^1_a$} (cb);
    \path (bc) edge node[above] {$\delta^0_a$} (ac);
    \path (bc) edge node[above] {$\delta^1_a$} (cc);
    \path (ab) edge node[left] {$\delta^0_b$} (aa);
    \path (ab) edge node[left] {$\delta^1_b$} (ac);
    \path (bb) edge node[left] {$\delta^0_b$} (ba);
    \path (bb) edge node[left] {$\delta^1_b$} (bc);
    \path (cb) edge node[left] {$\delta^0_b$} (ca);
    \path (cb) edge node[left] {$\delta^1_b$} (cc);
    \path (bb) edge node[above left] {$\delta^1_{ab}\!\!$} (cc);
    \path (bb) edge node[above left] {$\delta^0_{ab}\!\!$} (aa);
    \node[below left] at (aa) {$\bot\;$};
    \node[below left] at (ab) {$\bot\;$};
    \node[above right] at (cb) {$\;\top$};
    \node[above right] at (cc) {$\;\top$};
    \node[above right] at (ab) {$\;\top$};
    \node[right] at (5,0.8) {$\bot_X=\{v, g\}$};
    \node[right] at (5,0) {$\top_X=\{h, y, g\}$};
    \begin{scope}[shift={(8.8,.4)}, x=1.3cm, y=1.3cm]
      \filldraw[color=black!20] (0,0)--(2,0)--(2,2)--(0,2)--(0,0);			
      \filldraw (0,0) circle (0.05);
      \filldraw (2,0) circle (0.05);
      \filldraw (0,2) circle (0.05);
      \filldraw (2,2) circle (0.05);
      \path[red!50!black,line width=1] (0,0) edge node[below, black] {$\vphantom{b}a$} (1.95,0);
      \path[red!50!black,line width=1] (0,2) edge (1.95,2);
      \path[green!50!black,line width=1] (0,0) edge node[pos=.6, left, black] {$\vphantom{bg}b$} (0,1.95);
      \path[green!50!black,line width=1] (2,0) edge (2,1.95);
      \node[left] at (0,0) {$\bot$};
      \node[left] at (0,0.7) {$\bot$};
      \node[right] at (0,0.7) {$\top$};
      \node[right] at (2,2) {$\top$};
      \node[right] at (2,1) {$\top$};
      
      \node[blue!70,centered] at (0,-0.2) {$v$};
      \node[centered, red!50!black] at (1,0.15) {$e$};
      \node[blue!70,centered] at (2,-0.2) {$w$};
      \node[centered,blue!70] at (2,2.2) {$y$};
      \node[centered,blue!70] at (0,2.2) {$x$};
      \node[centered, green!50!black] at (0.2,1.1) {$\vphantom{bg}g$};
      \node[centered, green!50!black] at (1.8,1.1) {$\vphantom{bg}h$};
      \node[centered] at (1,1) {$q$};
      \node[centered, red!50!black] at (1,1.75) {$f$};
    \end{scope}
  \end{tikzpicture}
  \caption{A two-dimensional HDA $X$ on $\Sigma=\{a, b\}$, see Example~\ref{ex:hda}.}
  \label{fi:abcube}
\end{figure}

A \emph{higher-dimensional automaton} (\emph{HDA}) $X=(X, \bot_X, \top_X)$
is a precubical set together with subsets $\bot_X, \top_X\subseteq X$
of \emph{start} and \emph{accept} cells.
While HDAs may have an infinite number of cells, we will mostly be interested in finite HDAs.
Thus, in the following we will omit the word ``finite'' and will be explicit when talking about infinite HDAs.
The \emph{dimension} of an HDA $X$ is $\dim(X)=\sup\{|\ev(x)|\mid x\in X\}\in \Nat\cup\{\infty\}$.

\begin{remark}
  Precubical sets are presheaves over a category on objects $\square$,
  and then HDAs form a category with the induced morphisms,
  see \cite{DBLP:conf/concur/FahrenbergJSZ22}.
\end{remark}

For a precubical set $X$ and $k\ge 0$ we write
$X_{\le k}=\{x\in X\mid |\ev(x)|\le k\}$ for its \emph{$k$-skeleton},
\ie its restriction to cells of dimension at most $k$.
It is clear that this is again a precubical set,
and we use the same notation for HDAs.

A standard automaton is the same as a one-dimensional HDA $X=X_{\le 1}$
with the property that for all $x \in \bot_X \cup \top_X$, $\ev(x) = \emptyset$:
cells in $X[\emptyset]$ are states,
cells in $X[\{a\}]$ for $a\in \Sigma$ are $a$-labeled transitions,
and face maps $\delta_{\{a\}}^0$ and $\delta_{\{a\}}^1$
attach source and target states to transitions.
In contrast to ordinary automata, one-dimensional HDAs may have start and accept \emph{transitions}
instead of merely states,
so languages of one-dimensional HDAs may contain words with interfaces.

\begin{example}
  \label{ex:hda}
  Figure~\ref{fi:abcube} shows a two-dimensional HDA as a combinatorial object (left)
  and in a geometric realisation (right).
  It consists of
  nine cells: 
  the corner cells $X_0 = \{x,y,v,w\}$ in which no event is active (for all $z \in X_0$, $\ev(z) = \emptyset$),
  the transition cells $X_1 = \{g,h,f,e\}$ in which one event is active ($\ev(f) = \ev(e) = a$ and $\ev(g) = \ev(h) = b$),
  and the square cell $q$ where $\ev(q) = \loset{a\\b}$.

  The arrows between the cells on the left representation correspond to the face maps connecting them.
  For example, the upper face map $\delta^1_{a b}$ maps $q$ to $y$
  because the latter is the cell in which the active events $a$ and $b$ of $q$ have been terminated.
  On the right, face maps are used to glue cells together,
  so that for example $\delta^1_{a b}(q)$ is glued to the top right of $q$.
  In this and other geometric realisations,
  when we have two concurrent events $a$ and $b$ with $a\evord b$, we will draw $a$ horizontally and $b$ vertically.
\end{example}

\subsection{Regular languages}

Computations of HDAs are \emph{paths}, \ie sequences
\begin{equation*}
  \alpha=(x_0, \phi_1, x_1, \dotsc, x_{n-1}, \phi_n, x_n)
\end{equation*}
consisting of cells $x_i$ of $X$ and symbols $\phi_i$ which indicate face map types:
for every $i\in\{1,\dotsc, n\}$, $(x_{i-1}, \phi_i, x_i)$ is either
\begin{itemize}
\item $(\delta^0_A(x_i), \arrO{A}, x_i)$ for $A\subseteq \ev(x_i)$ (an \emph{upstep})
\item or $(x_{i-1}, \arrI{A}, \delta^1_A(x_{i-1}))$ for $A\subseteq \ev(x_{i-1})$ (a \emph{downstep}).
\end{itemize}
Downsteps terminate events, following upper face maps,
whereas upsteps start events by following inverses of lower face maps.
Both types of steps may be empty, and ${\arrO{\emptyset}}={\arrI{\emptyset}}$.

The \emph{source} and \emph{target} of $\alpha$ as above are $\src(\alpha)=x_0$ and $\tgt(\alpha)=x_n$.
The set of all paths in $X$ starting at $Y\subseteq X$ and terminating in $Z\subseteq X$
is denoted by $\Path(X)_Y^Z$.
A path $\alpha$ is \emph{accepting} if $\src(\alpha)\in \bot_X$ and $\tgt(\alpha)\in \top_X$.
Paths $\alpha$ and $\beta$ may be concatenated
if $\tgt(\alpha)=\src(\beta)$.
Their concatenation is written $\alpha*\beta$ or simply $\alpha \beta$.

\emph{Path equivalence} is the congruence $\simeq$
generated by $(z\arrO{A} y\arrO{B} x)\simeq (z\arrO{A\cup B} x)$,
$(x\arrI{A} y\arrI{B} z)\simeq (x\arrI{A\cup B} z)$, and
$\gamma \alpha \delta\simeq \gamma \beta \delta$ whenever $\alpha\simeq \beta$.
Intuitively, this relation allows to assemble subsequent upsteps or downsteps into one bigger step.
A path is \emph{sparse} if its upsteps and downsteps are alternating,
so that no more such assembling may take place.
Every equivalence class of paths contains a unique sparse path.

The observable content or \emph{event ipomset} $\ev(\alpha)$
of a path $\alpha$ is defined recursively as follows:
\begin{itemize}
\item if $\alpha=(x)$, then
  $\ev(\alpha)=\id_{\ev(x)}$;
\item if $\alpha=(y\arrO{A} x)$, then
  $\ev(\alpha)=\starter{\ev(x)}{A}$;
\item if $\alpha=(x\arrI{A} y)$, then
  $\ev(\alpha)=\terminator{\ev(x)}{A}$;
\item if $\alpha=\alpha_1*\dotsm*\alpha_n$ is a concatenation, then
  $\ev(\alpha)=\ev(\alpha_1)*\dotsm*\ev(\alpha_n)$.
\end{itemize}
Note that upsteps in $\alpha$ correspond to starters in $\ev(\alpha)$ and downsteps correspond to terminators.
Path equivalence $\alpha\simeq \beta$ implies $\ev(\alpha)=\ev(\beta)$~\cite{DBLP:conf/concur/FahrenbergJSZ22}.
Further, if $\alpha=\alpha_1*\dotsm*\alpha_n$ is a sparse path,
then $\ev(\alpha)=\ev(\alpha_1)*\dotsm*\ev(\alpha_n)$ is a sparse step decomposition.

The \emph{language} of an HDA $X$ is
$\Lang(X) = \{\ev(\alpha)\mid \alpha \text{ accepting path in } X\}$.

\begin{example}
  \label{ex:paths}
  The HDA $X$ of Figure~\ref{fi:abcube} admits several accepting paths with target $h$,
  for example $v\arrO{ab} q\arrI{a} h$.
  This is a sparse path and equivalent to the non-sparse paths
  $v\arrO{a} e\arrO{b} q\arrI{a} h$ and $v\arrO{b} g\arrO{a} q\arrI{a} h$.
  Their event ipomset is $\loset{a\\b\ibullet}$.
  In addition, since $g$ is both a start and accept cell,
  we have also $g$ and $v\arrO{b} g$ as accepting paths, with event ipomsets $\ibullet b \ibullet$ and $b \ibullet$, respectively.
  We have $\Lang(X)=\{
  b\ibullet,
  \ibullet b\ibullet,
  \loset{a\\b\ibullet}, 
  \loset{a\\\ibullet b\ibullet},
  \loset{a\\b},
  \loset{a\\\ibullet b}
  \}\down$.
\end{example}

The following property, that languages of skeleta are width restrictions of languages, is clear.

\begin{lemma}
  \label{le:skeleton}
  For any HDA $X$ and $k\ge 0$, $\Lang(X_{\le k})=\Lang(X)_{\le k}$.
\end{lemma}

In the lemma below, we write $\Path(X)_Y=\Path(X)_Y^X$, $\Path(X)^Z=\Path(X)_X^Z$, and $\Path(X)=\Path(X)_X^X$.

\begin{lemma}[\cite{DBLP:journals/corr/abs-2210-08298}]
  \label{l:PathDivision}
  Let $X$ be an HDA, $x,y\in X$ and $\gamma\in\Path(X)_x^y$. 
  Assume that $\ev(\gamma)=P*Q$ for ipomsets $P$ and $Q$.
  Then there exist paths $\alpha\in\Path(X)_x$ and $\beta\in\Path(X)^y$
  such that $\ev(\alpha)=P$, $\ev(\beta)=Q$ and $\tgt(\alpha)=\src(\beta)$.
\end{lemma}

\begin{lemma}
  \label{l:GeneralPathDivision}
  Let $X$ be an HDA, $P\in \Lang(X)$ and $P=P_1*\dotsm*P_n$ be any decomposition (not necessarily a step decomposition).
  Then there exists an accepting path $\alpha=\alpha_1*\dotsm*\alpha_n$ in $X$
  such that $\ev(\alpha_i)=P_i$ for all $i$.
  If $P=P_1*\dotsm*P_n$ is a sparse step decomposition, then $\alpha=\alpha_1*\dotsm*\alpha_n$ is sparse.
\end{lemma}

\begin{proof}
  The first claim follows from Lemma~\ref{l:PathDivision} by induction.
  As to the second, if starters and terminators are alternating in $P_1*\dotsm*P_n$,
  then upsteps and downsteps are alternating in $\alpha_1*\dotsm*\alpha_n$.
\end{proof}

Languages of HDAs are sets of (interval) ipomsets which are closed under subsumption~\cite{DBLP:conf/concur/FahrenbergJSZ22},
\ie languages in our sense.
A language is \emph{regular} if it is the language of a finite HDA.
  
\begin{theorem}[\cite{DBLP:conf/concur/FahrenbergJSZ22}]
  \label{th:kleene}
  A language is regular iff it is rational.
\end{theorem}

\begin{remark}
  Every ipomset $P$ may be converted into a \emph{track object} $\square^P$,
  see~\cite{DBLP:conf/concur/FahrenbergJSZ22},
  which is an HDA with the property that
  for any HDA $X$, $P\in \Lang(X)$ iff there is a morphism $\square^P\to X$.
  This is notably useful in proofs, for example of Lemma~\ref{l:PathDivision} above.
\end{remark}

\section{Regular and non-regular languages}
\label{se:stuff}

We exhibit a property, similar to the pumping lemma for word languages, which allows to show that some languages are non-regular.
Afterwards we show that intersections of regular languages are again regular.

\subsection{Pumping lemma}

\begin{lemma}
  \label{le:localpump}
  Let $L$ be a regular language.
  There exists $k\in \Nat$ such that for any $P\in L$, 
  any decomposition $P=Q_1*\dotsm*Q_n$
  with $n>k$ and any $0\le m\le n-k$
  there exist $i,j$ such that $m\le i< j\le m+k$
  and $Q_1*\dotsm*Q_{i}*(Q_{i+1}*\dotsm*Q_{j})^+*Q_{j+1}*\dotsm*Q_n\subseteq L$.
\end{lemma}

\begin{proof}
  Let $X$ be an HDA accepting $L$ and $k > |X|$.
  By Lemma~\ref{l:GeneralPathDivision}
  there exists an accepting path $\alpha= \alpha_1 * \dotsm * \alpha_n$
  such that $\ev(\alpha_i) = Q_i$ for all $i$, and $\ev(\alpha) =P$.
  Denote $x_i=\tgt(\alpha_i)=\src(\alpha_{i+1})$.
  Amongst the cells $x_m,\dotsc,x_{m+k}$ there are at least two equal,
  say $x_i=x_j$, $m\le i<j\le m+k$.
  As a consequence, $\src(\alpha_{i+1})=\tgt(\alpha_j)$,
  and for every $r\ge 1$,
  \begin{equation*}
    \alpha_1 * \dotsm * \alpha_i * (\alpha_{i+1}*\dotsm *\alpha_j)^r*\alpha_{j+1}*\dotsm*\alpha_n
  \end{equation*}
  is an accepting path that recognizes
  \begin{equation*}
    Q_1*\dotsm*Q_{i}*(Q_{i+1}*\dotsm*Q_{j})^r*Q_{j+1}*\dotsm*Q_n.
  \end{equation*}
  \qed
\end{proof}

\begin{corollary}
  Let $L$ be a regular language.
  There exists $k\in \Nat$ such that any $P\in L$ with $\size(P) > k$ 
  can be decomposed into $P=Q_1*Q_2*Q_3$
  such that $Q_2$ is not an identity and $Q_1*Q_2^+*Q_3\subseteq L$.
\end{corollary}

The proof follows by applying Lemma~\ref{le:localpump} to a dense step decomposition $P=Q_1*\dots* Q_{2\, \size(P)}$,
\cf Lemma~\ref{le:ipomdense}.
We may now expose a language which is not regular.

\begin{proposition}
  \label{pr:lnotreg}
  The language $L=\{\loset{a\\a}^n*a^n\mid n\ge 1\}\down$ is not regular.
\end{proposition}

Note that the restriction
$L_{\le 1}=(a a a)^+$
\emph{is} regular,
showing that regularity of languages may not be decided by restricting to their one-dimensional parts.

\begin{proof}[of Proposition~\ref{pr:lnotreg}]
  We give two proofs.
  The first uses Theorem~\ref{th:MN}:
  for every $k\ge 1$,
  $\loset{a\\a}^k\backslash L=\{\loset{a\\a}^n*a^{n+k}\mid n\ge 0\}\down$,
  and these are different for different $k$, so $\suff(L)$ is infinite.

  The second proof uses Lemma~\ref{le:localpump}.
  Assume $L$ to be regular,
  let $k$ be the constant from the lemma,
  and take $P=\loset{a\\a}^k*a^k=Q_1*\dots*Q_k*Q_{k+1}$,
  where $Q_1=\dotsm=Q_k=\loset{a\\a}$ and $Q_{k+1}=a^k$.
  For $m=0$ we obtain that $\loset{a\\a}^{k+(j-i)r}a^k\in L$
  for all $r$ and some $j-i>0$: a contradiction.
  \qed
\end{proof}

We may strengthen the above result
to show that regularity of languages may not be decided by restricting to their $k$-dimensional parts for any $k\ge 1$.
For $a\in \Sigma$ let $a^{\|_1}=a$ and $a^{\|_k}=a\parallel a^{\|_{k-1}}$ for $k\ge 2$:
the $k$-fold parallel product of $a$ with itself.
Now let $k\ge 1$ and
\begin{equation*}
  L = \big\{ (a^{\|_{k+1}})^n * P^n \bigmid n\ge 0, P\in \{a^{\|_{k+1}}\}\down - \{ a^{\|_{k+1}}\} \big\}\down.
\end{equation*}
The idea is to remove from the right-hand part of the expression
precisely the only ipomset of width $k+1$.
Using the same arguments as above one can show that $L$ is not regular,
but $L_{\le k}=((\{a^{\|_{k+1}}\}\down - \{ a^{\|_{k+1}}\})^2)^+$ is.

Yet the $k$-restrictions of any regular language remain regular:

\begin{proposition}
  \label{prop:regrestriction}
  Let $k \ge 0$. If $L\in \Langs$ is regular, then so is $L_{\le k}$.
\end{proposition}

\begin{proof}
  By Lemma \ref{le:skeleton}. \qed
\end{proof}

\subsection{Intersection of regular languages}
\label{se:intersection}

By definition,
the regular languages are closed under union, parallel composition, gluing composition, and Kleene plus.
Here we show that they are also closed under intersection.
(For complement this is more complicated, as we will see later.)

\begin{proposition}
  \label{pr:cap}
  The regular languages are closed under $\cap$.
\end{proposition}

\begin{proof}
  We again give two proofs,
  one algebraic using Theorem~\ref{th:MN} and another, constructive proof using Theorem~\ref{th:kleene}.
  For the first proof,
  let $L_1$ and $L_2$ be regular, then $\suff(L_1)$ and $\suff(L_2)$ are both finite.
  Now
  \begin{align*}
    \suff(L_1&\cap L_2) \\
    &= \{P\backslash(L_1\cap L_2)\mid P\in \iiPoms\} \\
    &= \big\{ \{Q\in \iiPoms\mid P Q\in L_1\cap L_2\} \bigmid P\in \iiPoms\big\} \\
    &= \big\{ \{Q\in \iiPoms\mid P Q\in L_1\}\cap \{Q\in \iiPoms\mid P Q\in L_2\} \bigmid P\in \iiPoms\big\} \\
    &= \{P\backslash L_1 \cap P\backslash L_2\mid P\in \iiPoms\} \\
    &\subseteq \{M_1\cap M_2 \bigmid M_1\in\suff(L_1),\;M_2\in\suff(L_2)  \}
  \end{align*}
  which is thus finite.

  For the second, constructive proof,
  let $X_1$ and $X_2$ be HDAs.
  We construct an HDA $X$ with $\Lang(X)=\Lang(X_1)\cap \Lang(X_2)$:
  \begin{equation}
    \label{eq:product}
    \begin{gathered}
      X = \{(x_1,x_2) \in X_1 \times X_2 \mid \ev_1(x_1) = \ev_2(x_2)\}, \\
      \delta_A^\nu(x_1,x_2) = (\delta_A^{\nu}(x_1), \delta_A^{\nu}(x_2)), \\
      \ev((x_1,x_2)) = \ev_1(x_1) = \ev_2(x_2), \\
      \bot  = \bot_1 \times \bot_2, \qquad
      \top  = \top_1 \times \top_2.
    \end{gathered}
  \end{equation}

  For the inclusion $\Lang(X)\subseteq \Lang(X_1)\cap \Lang(X_2)$,
  any accepting path $\alpha$ in $X$
  projects to accepting paths $\beta$ in $X_1$ and $\gamma$ in $X_2$,
  and then $\ev(\beta)=\ev(\gamma)=\ev(\alpha)$.
  For the reverse inclusion, we need to be slightly more careful
  to ensure that accepting paths in $X_1$ and $X_2$ may be assembled
  to an accepting path in $X$.
  
  Let $P\in \Lang(X_1)\cap \Lang(X_2)$
  and $P=P_1*\dotsm*P_n$ the sparse step decomposition.
  Let $\beta=\beta_1*\dotsm*\beta_n$ and $\gamma=\gamma_1*\dotsm*\gamma_n$
  be sparse accepting paths for $P$ in $X_1$ and $X_2$, respectively,
  such that $\ev(\alpha_i)=\ev(\beta_i)=P_i$ for all $i$, \cf Lemma~\ref{l:GeneralPathDivision}.

  Let $i\in\{1,\dotsc, n\}$ and assume that $P_i=\starter{U}{A}$ is a starter,
  then $\beta_i=(\delta_A^0 x_1, \arrO{A}, x_1)$
  and $\gamma_i=(\delta_A^0 x_2, \arrO{A}, x_2)$
  for $x_1\in X_1$ and $x_2\in X_2$ such that $\ev(x_1)=\ev(x_2)=U$.
  Hence we may define a step $\alpha_i=(\delta_A^0(x_1, x_2), \arrO{A}, (x_1, x_2))$ in $X$.
  If $P_i$ is a terminator, the argument is similar.
  By construction, $\tgt(\alpha_i)=\src(\alpha_{i+1})$,
  so the steps $\alpha_i$ assemble to an accepting path $\alpha = \alpha_1 * \dots * \alpha_n \in \Path(X)_\bot^\top$,
  and $\ev(\alpha)=P$.
  \qed
\end{proof}

\begin{remark}
  The HDA $X$ constructed in \eqref{eq:product} above is the product in the category of HDAs.
  This gives a third and more high-level proof of Proposition \ref{pr:cap},
  using track objects and the universal property of the product.
\end{remark}

\section{ST-automata}
\label{se:st}

In the following, we define \emph{ST-automata} whose languages are sets of words over an alphabet of starters and terminators.
We use here the most recent definition of these objects introduced in \cite{conf/ramics/AmraneBCFZ24}.
Variants of ST-automata have been used in  \cite{DBLP:conf/concur/FahrenbergJSZ22,
	DBLP:conf/ictac/AmraneBFZ23,
	DBLP:conf/apn/AmraneBCF24,
	DBLP:journals/lites/Fahrenberg22,
	DBLP:conf/adhs/Fahrenberg18}.
We use a construction from~\cite{conf/ramics/AmraneBCFZ24}
which translates HDAs into ST-automata.
This will be useful for showing properties of HDA languages.

Let us first introduce some notation.
Let $\Omega=\{\starter{U}{A}, \terminator{U}{A}\mid U\in \square, A\subseteq U\}$
be the (infinite) set of starters and terminators over $\Sigma$
and $\Id=\{\id_U\mid U\in \square\} \subseteq \Omega$ be the (infinite) set of identities.
Further, for any $k\ge 0$, let $\Omega_{\le k}=\Omega\cap \iiPoms_{\le k}$
and $\Id_{\le k}=\Id\cap \iiPoms_{\le k}$.
Note that both these sets are finite and $\Id_{\le k} \subseteq \Omega_{\le k}$.

An \emph{ST-automaton} is a structure $A=(Q, E, I, F, \lambda)$
consisting of sets $Q$, $E\subseteq Q\times \Omega\times Q$, $I, F\subseteq Q$,
and a function $\lambda: Q\to \square$ such that
for all $(q, \ilo{S}{U}{T}, r)\in E$, $\lambda(q)=S$ and $\lambda(r)=T$. 
Thus, an ST-automaton is a standard automaton over $\Omega$ where the states are also labeled by $\square$ consistently with the labeling of incoming and outgoing edges.

A \emph{path} in an ST-automaton $A$ is defined as
an alternating sequence $\pi=(q_0, e_1, q_1,\dots, e_n, q_n)$
of states $q_i$ and transitions $e_i$ such that $e_i=(q_{i-1}, P_i, q_i)$ for  every $i=1,\dots, n$
and some $P_1,\dots, P_n\in \Omega$.
The path is \emph{accepting} if $q_0\in I$ and $q_n\in F$.
The \emph{label} of $\pi$ is
$\ell(\pi)=\id_{\lambda(q_0)} P_1 \id_{\lambda(q_1)}\dots P_n \id_{\lambda(q_n)}$.
The \emph{language} of an ST-automaton $A$ is
\begin{equation*}
	\Lang(A) = \{\ell(\pi)\mid \pi \text{ accepting path in } A\}.
\end{equation*}

Note that languages of ST-automata are subsets of $\Cohnew$.
In particular, the labeling of states and their consideration in the path labels forbid to have the empty word $\epsilon$ in the language of an ST-automata.

To a given HDA $X=(X, \bot, \top)$ we associate an ST-automaton $\ST(X)=(Q, E, I, F, \lambda)$ as follows:
\begin{itemize}
	\item $Q=X$, $I=\bot$, $F=\top$, $\lambda=\ev$, and
	\item $E=\{(\delta_A^0(q), \starter{\ev(q)}{A}, q)\mid A\subseteq \ev(q)\}
	\cup \{(q, \terminator{\ev(q)}{A}, \delta_A^1(q))\mid A\subseteq \ev(q)\}$.
\end{itemize}
That is, the transitions of $\ST(X)$ precisely mimic the starting and terminating of events in $X$.
Note that given an HDA $X$ of dimension $\dim(X) = k \in \Nat$, $\ST(X)$ is over the finite alphabet $\Omega_{\le k}$.
Note also that lower faces in $X$ are inverted to get the starting transitions.

In what follows, we consider languages of nonempty words over $\Omega$, which we denote by $W$ etc.\
and the class of such languages by $\Wangs$.
Further, $\Wang(\mathcal{A})$ denotes the set of words
accepted by a finite automaton $\mathcal{A}$.

\begin{figure}[tbp]
  \centering
   \begin{tikzpicture}[y=.9cm, scale=0.9, every node/.style={transform shape}]
   \centering
   \node[circle,draw=black,fill=blue!30,inner sep=0pt,minimum size=15pt]
    (aa) at (0,0) {$\vphantom{hy}v$};
    \node[circle,draw=black,fill=blue!30,inner sep=0pt,minimum size=15pt]
    (ac) at (0,4) {$\vphantom{hy}x$};
    \node[circle,draw=black,fill=blue!30,inner sep=0pt,minimum size=15pt]
    (ca) at (4,0) {$\vphantom{hy}w$};
    \node[circle,draw=black,fill=blue!30,inner sep=0pt,minimum size=15pt]
    (cc) at (4,4) {$\vphantom{hy}y$};
    \node[circle,draw=black,fill=red!30,inner sep=0pt,minimum size=15pt]
    (ba) at (2,0) {$\vphantom{hy}e$};
    \node[circle,draw=black,fill=red!30,inner sep=0pt,minimum size=15pt]
    (bc) at (2,4) {$\vphantom{hy}f$};
    \node[circle,draw=black,fill=green!30,inner sep=0pt,minimum size=15pt]
    (ab) at (0,2) {$\vphantom{hy}g$};
    \node[circle,draw=black,fill=green!30,inner sep=0pt,minimum size=15pt]
    (cb) at (4,2) {$\vphantom{hy}h$};
    \node[circle,draw=black,fill=black!10,inner sep=0pt,minimum size=15pt]
    (bb) at (2,2) {$\vphantom{hy}q$};
    \path (ba) edge node[above] {$\delta^0_a$} (aa);
    \path (ba) edge node[above] {$\delta^1_a$} (ca);
    \path (bb) edge node[above] {$\delta^0_a$} (ab);
    \path (bb) edge node[above] {$\delta^1_a$} (cb);
    \path (bc) edge node[above] {$\delta^0_a$} (ac);
    \path (bc) edge node[above] {$\delta^1_a$} (cc);
    \path (ab) edge node[left] {$\delta^0_b$} (aa);
    \path (ab) edge node[left] {$\delta^1_b$} (ac);
    \path (bb) edge node[left] {$\delta^0_b$} (ba);
    \path (bb) edge node[left] {$\delta^1_b$} (bc);
    \path (cb) edge node[left] {$\delta^0_b$} (ca);
    \path (cb) edge node[left] {$\delta^1_b$} (cc);
    \path (bb) edge node[above left] {$\delta^1_{ab}\!\!$} (cc);
    \path (bb) edge node[above left] {$\delta^0_{ab}\!\!$} (aa);
    \node[below left] at (aa) {$\bot\;$};
    \node[below left] at (ab) {$\bot\;$};
    \node[above right] at (ab) {$\;\top$};
    \node[above right] at (cb) {$\;\top$};
    \node[above right] at (cc) {$\;\top$};
    \begin{scope}[shift={(7,0)},initial text =]
      \tikzstyle{corner} = [circle,draw=black,fill=blue!30,inner sep=0pt,minimum size=15pt]
    \node[corner,initial]
    (aa) at (0,0) {$\vphantom{hy}\emptyset$};
    \node[corner]
    (ac) at (0,4) {$\vphantom{hy}\emptyset$};
    \node[corner]
    (ca) at (4,0) {$\vphantom{hy}\emptyset$};
    \node[corner, accepting by double]
    (cc) at (4,4) {$\vphantom{hy}\emptyset$};
    \node[corner,fill=red!30]
    (ba) at (2,0) {$\vphantom{hy}a$};
    \node[corner,fill=red!30]
    (bc) at (2,4) {$\vphantom{hy}a$};
    \node[corner, accepting by double,fill=green!30,initial]
    (ab) at (0,2) {$\vphantom{hy}b$};
    \node[corner, accepting by double,fill=green!30]
    (cb) at (4,2) {$\vphantom{hy}b$};
    \node[corner,fill=black!10]
    (bb) at (2,2) {$\vphantom{hy}\loset{a \\ b}$};
    \path (aa) edge node[above] {$a\ibullet$} (ba);
    \path (ba) edge node[above] {$\ibullet a$} (ca);
    \path (ab) edge node[above] {$\loset{\phantom{\ibullet} a \ibullet \\ \ibullet b \ibullet }$} (bb);
    \path (bb) edge node[below] {$\loset{\ibullet a\phantom{\ibullet}\\ \ibullet b \ibullet }$} (cb);
    \path (ac) edge node[above] {$a \ibullet$} (bc);
    \path (bc) edge node[above] {$ \ibullet a$} (cc);
    \path (aa) edge node[left] {$b \ibullet$} (ab);
    \path (ab) edge node[left] {$\ibullet b$} (ac);
    \path (ba) edge node[right] {$\loset{\ibullet a \ibullet \\ \phantom{\ibullet} b \ibullet }$} (bb);
    \path (bb) edge node[left] {$\loset{\ibullet a \ibullet \\ \ibullet b \phantom{\ibullet}}$} (bc);
    \path (ca) edge node[left] {$b \ibullet$} (cb);
    \path (cb) edge node[left] {$\ibullet b$} (cc);
    \path (bb) edge node[above left=-0.15cm] {$\loset{\ibullet a \\ \ibullet b}$} (cc);
    \path (aa) edge node[above left=-0.15cm] {$\loset{a \ibullet \\ b \ibullet }$} (bb);
    \end{scope}
  \end{tikzpicture}
  \caption{HDA of Figure~\ref{fi:abcube} and its ST-automaton
    (identity loops not displayed).}
  \label{fi:abst}
\end{figure}

\begin{example}
  \label{ex:stautomata}
  Figure~\ref{fi:abst} displays the ST-automaton $\ST(X)$
  pertaining to the HDA $X$ in Figure~\ref{fi:abcube},
  with the identity loops $(z, \id_{\ev(z)}, z)$ for all states $z$ omitted.
  Notice that the transitions between a cell and its lower face are opposite to the face maps in $X$.
  The smallest accepting path of $\ST(X)$ is $\omega = (g)$ where $\lambda(g) = b$. Its labeling is $\ell(\omega) = \ibullet b \ibullet$.
  More generally,
  $\Wang(\ST(X))=\{
  \ibullet b \ibullet,
  \id_\emptyset b\ibullet \ibullet b \ibullet,
  \id_\emptyset \loset{a\ibullet\\b\ibullet} \loset{\ibullet a \ibullet \\ \ibullet b \ibullet} \loset{\ibullet a\\\ibullet b\ibullet}\ibullet b \ibullet,
  \ibullet b\ibullet \loset{a\ibullet\\\ibullet b\ibullet} \loset{\ibullet a \ibullet \\ \ibullet b \ibullet} \loset{\ibullet a\\\ibullet b\ibullet},
  \id_\emptyset \loset{a\ibullet \\ b \ibullet} \loset{\ibullet a \ibullet \\ \ibullet b \ibullet} \loset{\ibullet a \\ \ibullet b} \id_\emptyset
  ,\dotsc\}$.
\end{example}

Define functions $\Phi: \Langs\to \Wangs$ and $\Psi: \Wangs\to \Langs$ by
\begin{align*}
  \Phi(L) &= \{P_1\dotsm P_n\in \Omega^*\mid  P_1*\dotsm*P_n\in L,\; n\ge 1,\; \forall i: P_i\in \Omega\}, \\
  \Psi(W) &= \{P_1*\dotsm*P_n\in \iiPoms \mid P_1\dotsm P_n\in W,\;n\ge 1,\; \forall i: T_{P_i}=S_{P_{i+1}}\}\down.
\end{align*}

$\Phi$ translates ipomsets into concatenations of their step decompositions,
and $\Psi$ translates words of composable starters and terminators into their ipomset composition
(and takes subsumption closure).
Hence $\Phi$ creates ``coherent'' words, \ie nonempty concatenations of starters and terminators with matching interfaces.
Conversely, $\Psi$ disregards all words which are not coherent in that sense.
Every ipomset is mapped by $\Phi$
to infinitely many words over $\Omega$
(because ipomsets $\id_U\in\Omega$ are not units in $\Wangs$).
This will not be a problem for us later.
It is clear that $\Psi(\Phi(L))=L$ for all $L\in \Langs$,
since every ipomset has a step decomposition.
For the other composition,
neither $\Phi(\Psi(W))\subseteq W$ nor $W\subseteq \Phi(\Psi(W))$ hold:

\begin{example}
  If $W = \{a\ibullet\, \ibullet b\}$
  (the word language containing the concatenation of $a\ibullet$ and $\ibullet b$),
  then $\Psi(W) = \emptyset$ and thus $\Phi(\Psi(W)) = \emptyset \not\supseteq W$.
  If $W = \{ \loset{a\ibullet \\ b\ibullet}\! \loset{\ibullet a \\ \ibullet b} \}$,
  then $\Psi(W) = \{\loset{a\\b}, ab, ba\}$
  and $\Phi(\Psi(W)) \not\subseteq W$.
\end{example}

\begin{lemma}
  \label{lenotgluinglanguages}
  \label{le:notgluinglanguages}
  For all $A_1,A_2 \subseteq \iiPoms$, $A_1\down * A_2\down = \{P_1*P_2 \mid P_1\in A_1,~P_2\in A_2\}\down$.
\end{lemma}

\begin{proof}
  Let $R \in A_1\down * A_2\down$.
  By definition, there exists $P'_i\in A_i\down$ such that $R \subsu P'_1 * P'_2$. Let $P_i \in A_i$ such that $P'_i \subsu P_i$.
  Then $R \subsu P_1 * P_2$.
  The other inclusion follows from the facts that $A_i \subseteq A_i\down$ and that the gluing composition preserves subsumption. \qed
\end{proof}
 As  trivial consequences of the definitions, we have that   $\Phi$ respects boolean operations:

\begin{lemma}
  \label{le:phibool}
  For all $L_1, L_2\in \Langs$,
  $\Phi(L_1\cap L_2)= \Phi(L_1)\cap \Phi(L_2)$ and
  $\Phi(L_1\cup L_2)= \Phi(L_1)\cup \Phi(L_2)$.
\end{lemma}

On the other hand, $\Phi$ does \emph{not} respect concatenations:
only inclusion $\Phi(L*L')\subseteq \Phi(L)*\Phi(L')$ holds,
given that $\Phi(L)*\Phi(L')$ also may contain words in $\Omega^*$
that are not composable in $\iiPoms$.

As to $\Psi$, we show that it respects regular operations:

\begin{lemma}
	\label{le:psireg}
	For all $W_1, W_2\in \Wangs$,
	$\Psi(W_1\cup W_2)=\Psi(W_1)\cup \Psi(W_2)$,
	$\Psi(W_1 W_2)=\Psi(W_1)*\Psi(W_2)$, and
	$\Psi(W_1^+)=\Psi(W_1)^+$.
\end{lemma}	

\begin{proof}
  The first claim follows easily using the fact that $(A\cup B)\down=A\down\cup B\down$.
  For the second, we have
  \begin{align*}
    \Psi(W_1)*\Psi(W_2)
    &= \{P_1*\dotsm *P_n\mid P_1\dotsm P_n\in W_1,\; \forall i: T_{P_i}=S_{P_{i+1}}\}\down \\
    &\qquad * \{Q_1*\dotsm*Q_m\mid Q_1\dotsm Q_m\in W_2,\; \forall i: T_{Q_i}=S_{Q_{i+1}}\}\down \\
    &= \{P_1*\dotsm*P_n * P_{n+1}*\dotsm *P_{n+m}\mid P_1\dotsm P_n\in W_1, \\
    &\qquad P_{n+1}\dotsm P_{n+m}\in W_2,\; \forall i: T_{P_i}=S_{P_{i+1}}\}\down \text{ (by Lemma \ref{le:notgluinglanguages})}\\
    &= \Psi(W_1 W_2).
  \end{align*}
  The equality $\Psi(W_1^+) = \Psi(W_1)^+$ then follows by trivial recurrence,
  using the equalities for binary union and gluing composition. \qed
\end{proof}

On the other hand, $\Psi$ does not respect intersections,
given that $(A\cap B)\down=A\down \cap B\down$ does not always hold.

Let $\Cohnew\subseteq \Omega_{\le k}^*$ (which is a regular word language)
denotes the set of all words over $\Omega_{\le k}$ starting and ending with an identity and where identities and starters/terminators alternate. We have the following:

\begin{lemma}
  \label{le:philang}
  For any HDA $X$ and $k\ge \dim(X)$,
  \begin{equation*}
    \Wang(\ST(X))=\Phi(\Lang(X))\cap \Cohnew.
  \end{equation*}
\end{lemma}

\begin{proof}
  By construction, there is a one-to-one correspondence between the accepting paths in $X$ and $\ST(X)$:
  \[
    \alpha=
    (x_0,\phi_1,x_1, \phi_2,\dotsc, \phi_n, x_n)
    \mapsto
    \big(
      x_0 \xrightarrow{\psi_1}
      x_1\xrightarrow{\psi_2}
      \dotsm
      \xrightarrow{\psi_n}
      x_n
    \big)
    =\omega
  \]
  where $\psi_i$ is the starter or terminator corresponding to the step $\phi_i$ and $\lambda(x_i) = \id_{\ev(x_i)}$.
  Note also that the path $(x)$ is accepting in $X$ if and only if it is accepting in $\ST(X)$. In addition $\ev(x) = \ell(x)$.

  Now let $\id_0P_1\id_1\dotsm P_n\id_n\in \Wang(\ST(X))$ with $n>1$.
  Then there is an accepting path $\omega$ in $\ST(X)$ such that
   $P_i=\ev(x_{i-1},\phi_i,x_i) = \psi_i$, $\id_0 = \id_{\ev(x_{i - 1})}$ and $\id_i=\id_{\ev(x_{i})}$ for $1\le i \le n$.
  The corresponding path $\alpha$ in $X$ is accepting.
  Hence $P_1*\dotsm*P_n=\id_0*P_1*\id_1*\dotsm *P_n*\id_n=\ev(\alpha)\in\Lang(X)$,
  and $\id_0P_1\id_1\dotsm P_n\id_n\in \Phi(\Lang(X))$.
  This shows the inclusion $\subseteq$.

  Now let $\id_0P_1\id_1\dotsm P_n\id_n\in \Phi(\Lang(X))\cap \Cohnew$.
  Thus $\id_0*P_1*\id_1*\dotsm *P_n*\id_n\in\Lang(X)$.
  Using Lemma~\ref{l:GeneralPathDivision} we conclude that there exists an accepting
  path $\alpha=\src(\beta_1)*\beta_1*\tgt(\beta_1)*\dotsm*\src(\beta_n)*\beta_n*\tgt(\beta_n)$ in $X$ such that $\ev(\beta_i)=P_i$, $\ev(\src(\beta_i)) = \id_{i-1}$ and $\ev(\tgt(\beta_i) = \id_i$.
  The path $\omega$ corresponding to $\alpha$ recognizes $\id_0P_1\id_1\dotsm P_n\id_n$,
  which shows the inclusion $\supseteq$.
  \qed
\end{proof}

\begin{lemma}
  \label{le:inclphiiff}
  Let $k\ge 0$.
  For all languages $L_1, L_2\in \Langs_{\le k}$, $L_1\subseteq L_2$ if and only if $\Phi(L_1)\cap \Cohnew\subseteq \Phi(L_2)\cap \Cohnew$.
\end{lemma}

\begin{proof}
  The forward implication is immediate from Lemma~\ref{le:phibool}.
  Now if $L_1\not\subseteq L_2$,
  then also $\Phi(L_1)\cap \Cohnew\not\subseteq \Phi(L_2)\cap \Cohnew$,
  since every ipomset admits a step decomposition starting with an identity. \qed
\end{proof}

\begin{theorem}
  \label{th:incl}
  Inclusion of regular languages is decidable.
\end{theorem}

\begin{proof}
  Let $L_1$ and $L_2$ be regular and recognized respectively by $X_1$ and $X_2$,
  and let $k=\max(\dim(X_1), \dim(X_2))$.
  By Lemmas~\ref{le:philang} and~\ref{le:inclphiiff},
  \begin{align*}
    L_1\subseteq L_2
    & \iff
    \Phi(L_1)\cap \Cohnew\subseteq \Phi(L_2)\cap \Cohnew
    \\ &\iff
    \Wang(\ST(X_1))\subseteq \Wang(\ST(X_2)).
  \end{align*}
  Given that these are finite automata, the latter inclusion is decidable.
  \qed
\end{proof}

\section{Determinism and ambiguity}
\label{se:det}

It is shown in \cite{DBLP:journals/corr/abs-2210-08298} that not all regular languages may be recognized by deterministic HDAs.
We recall the definition from \cite{DBLP:journals/corr/abs-2210-08298}.
A cell $x\in X$ in an HDA $X$ is \emph{essential} if there exists an accepting path in $X$ that contains $x$.
A path is essential if all its cells are essential.

\begin{definition}
  An HDA $X$ is \emph{deterministic} if
  \begin{itemize}
  \item it has at most one start cell of every type:
    for all $U\in\square$,
    $|\bot_X\cap X[U]|\le 1$;
  \item for every $U\in\square$, $A\subseteq U$, and essential cell $x\in X[U-A]$, 
    there exists at most one cell $y\in X[U]$ such that $\delta_A(y)=x$.
  \end{itemize}
\end{definition}

We also say that a regular language is \emph{deterministic} if it is
recognized by a finite deterministic HDA.
The following is obvious.

\begin{lemma}
  \label{le:detskel}
  If $X$ is a deterministic HDA,
  then its $1$-skeleton $X_{\le 1}$ is a deterministic finite automaton.
\end{lemma}

\begin{lemma}[\cite{DBLP:journals/corr/abs-2210-08298}]
  \label{le:nodet}
  There exist regular languages which are not deterministic.
\end{lemma}

Hence there exist HDAs which may not be determinized.
By Lemma \ref{le:detskel} it is clear that any example of such a non-determinizable HDA must be at least two-dimensional.
The example below, taken from \cite{DBLP:journals/corr/abs-2210-08298}, provides such an HDA.

\begin{example}
  \label{ex:nonDet}
  Let $L = \{\loset{a \\ b}, abc\}\down = \{\loset{a \\ b}, ab, ba, abc\}$.
  Figure~\ref{fig:nondet} shows an  HDA $X$ which recognizes $L$.
  Note that $X$ is not deterministic since its bottom-right cell (the red one) has \emph{two} outgoing $b$-labeled edges.
  In Example~\ref{ex:nonDeterminizable} below we will see that $L$ itself is non-deterministic:
  it is not recognized by any deterministic HDA.
\end{example}

\begin{figure}
\centering
\begin{tikzpicture}[y=.9cm, scale=1.2, every node/.style={transform shape}]

		\filldraw[color=black!20] (0,0)--(2,0)--(2,2)--(0,2)--(0,0);			
		\filldraw (0,0) circle (0.05);
		\filldraw[color=red] (2,0) circle (0.05);
		\filldraw (0,2) circle (0.05);
		\filldraw (2,2) circle (0.05);
		
		\filldraw (3,1) circle (0.05);
		\path (0,0) edge node[below, black] {$\vphantom{b}a$} (1.95,0);
		\path (0,2) edge node[pos=.5, above, black] {$\vphantom{bg}a$} (1.95,2);
		\path (0,0) edge node[pos=.5, left, black] {$\vphantom{bg}b$} (0,1.95);
		\path (2,0) edge node[pos=.5, right, black] {$\vphantom{bg}b$} (2,1.95);
		\path (2,0) edge node[pos=.5,below right=-0.1cm, black] {$\vphantom{bg}b$} (3,0.95);
		\path (3,1) edge node[pos=.5,above right=-0.1cm, black] {$\vphantom{bg}c$} (2.05,2);
		\node[left] at (0,0) {$\bot$};
		\node[right] at (2,2.2) {$\top$};

\end{tikzpicture}
\caption{Non-deterministic HDA accepting the language of Example~\ref{ex:nonDet}.}
\label{fig:nondet}
\end{figure}

We will below prove a strengthening of Lemma \ref{le:nodet}, introducing a notion of \emph{ambiguity}
and showing that there exist regular languages with unbounded ambiguity.

\subsection{Decidability of determinism}

First we show that determinism is decidable.
To do so, we use a criterion developed in \cite{DBLP:journals/corr/abs-2210-08298} characterizing deterministic languages.

\begin{definition}
A language $L$ is \emph{swap-invariant} if it holds for all $P, Q, P'
, Q' \in  \iiPoms$ that $PP' \in L,
QQ' \in L$ and $P \subsu Q$ imply $QP' \in L$.
\end{definition}

\begin{lemma}[\cite{DBLP:journals/corr/abs-2210-08298}]
	A language $L$ is swap-invariant if and only if $P \subsu Q$ implies $P \backslash L = Q \backslash L$ for all $P, Q \in \iiPoms$	unless $Q \backslash L = \emptyset$.
\end{lemma}

\begin{theorem}[\cite{DBLP:journals/corr/abs-2210-08298}]
	A language $L$ is deterministic if and only if it is swap-invariant.
\end{theorem}

\begin{example}
	\label{ex:nonDeterminizable}
	Continuing Example~\ref{ex:nonDet}, note that $ab \subsu \loset{a \\ b}$ but $ab \backslash L = \{\epsilon,c\} \neq \{\epsilon\} = \loset{a \\ b} \backslash L$. Thus $L$ is not swap-invariant and hence not a deterministic regular language.
\end{example}

\begin{lemma}
  Determinism is preserved by intersection, but neither by union nor Kleene plus.
\end{lemma}

\begin{proof}
  To show preservation by intersection, let $L_1$ and $L_2$ be deterministic languages
  and $P, P', Q, Q'\in L=L_1\cap L_2$ such that $PP' \in L$, $QQ' \in L$ and $P \sqsupseteq Q$.
  Then, as $L_1$ and $L_2$ are swap-invariant, $QP' \in L_1$ and $QP' \in L_2$ and thus $QP' \in L$.
  We have shown that $L$ is swap invariant and thus deterministic.

  For union, consider $L_1 = \{\loset{a \\ a}aa \}$ and $L_2 = aa\loset{a\\a}\}$.
  Both $L_1$ and $L_2$ are deterministic, but $L_1 \cup L_2$ is not.

  As to Kleene plus, consider $L = \{b\ibullet, \loset{a \\ \ibullet b} c d, a b \loset{c\\d}\}$, then
  $L$ is deterministic,
  yet Proposition \ref{pr:lambig} below implies that $L^+$ is not.
  \qed
\end{proof}

In the following let $L$ be a regular language accepted by some HDA $X$.
Let 
\begin{equation*}
  \Pre(X) = \{\ev(\alpha) \mid \alpha=(x_0, \phi_1, \dotsc, \phi_n, x_n) \in \Path(X)_\bot, \forall i \ne j: x_i \neq x_j\},
\end{equation*}
that is the event ipomsets of all paths of $X$ starting at any initial cell and that do not loop back onto themselves.

\begin{lemma}
  \label{lem:swapInvariant}
  Let $P \in \iiPoms$. Then
  $P \backslash L \in \suff(L)$ if and only if there exists $Q \in \Pre(X)$
  such that $P \backslash L = Q \backslash L$.  
\end{lemma}

\begin{proof}
  The proof is immediate if $P \in \Pre(X)$.
  For the opposite it suffices to take $Q \in \Pre(X)$ whose path has the same target as the one of $P$.
\end{proof}

As corollaries we have:

\begin{corollary}
  \label{cor:suff=Pre}
  $\suff(L) = \{P \backslash L \mid P \in \Pre(X)\}$.
\end{corollary}

\begin{corollary}
  \label{cor:typesOfPrefixes}
  For all $P \in \iiPoms$ either $P \backslash L = \emptyset$ or there exists $Q \in \Pre(X)$ such that $P \backslash L = Q \backslash L$.
\end{corollary}

\begin{theorem}
  \label{th:deterministic}
  It is decidable whether $L$ is deterministic.
\end{theorem}

\begin{proof}
  First notice that the set $\Pre(X)$ is finite and computable.
  In addition, for all $P \in \Pre(X)$, $P \backslash L$ is accepted by the HDA $Y$ built as follows.  For all $x \in X$ such that there exists a non-looping $\alpha \in \Path(X)_\bot^x$ and $P = \ev(\alpha)$, let $(X_x,\{x\},\top_X)$ be a copy of $X$. Then $Y$ is the union of such $X_x$'s.
  We have now an HDA for each element of $\suff(L)$.
  Then, by Lemma~\ref{lem:swapInvariant} and Corollary~\ref{cor:typesOfPrefixes}, it suffices to check for all $P,Q \in \Pre(X)$ whether when $P\subsu Q$ we have $P\backslash L = Q \backslash L$.
  Since we can easily check subsumption the theorem follows from the decidability of inclusion (Theorem~\ref{th:incl}).
\end{proof}

\subsection{Ambiguity}

We now introduce a notion of ambiguity for HDAs and show that languages of unbounded ambiguity exist.

\begin{definition}
  Let $k\ge 1$.
  An HDA $X$ is \emph{$k$-ambiguous} if every
  $P\in \Lang(X)$ is the event ipomset of at most $k$ sparse accepting
  paths in $X$.
\end{definition}

A language $L$ is said to be \emph{of bounded ambiguity}
if it is recognized by a $k$-ambiguous HDA for some $k$.
Deterministic HDAs are $1$-ambiguous,
and deterministic regular languages are of bounded ambiguity.

\begin{proposition}
  \label{pr:lambig}
  The regular language $L = (\loset{a\\b} c d + a b \loset{c\\d})^+$
  is of unbounded ambiguity.
\end{proposition}

Before the proof, a lemma about the structure of accepting paths in any HDA which accepts $L$.

\begin{lemma}
  \label{l:BrickPaths}
  Let $X$ be an HDA with $\Lang(X)=L$.
  Let $\alpha$ and $\beta$ be essential sparse paths in $X$
  with $\ev(\alpha)=\loset{a\\b} c d$ and $\ev(\beta)=a b \loset{c\\d}$.
  Then
  \begin{align*}
    \alpha &= (v\arrO{ab}q\arrI{ab} x\arrO{c}e\arrI{c}y\arrO{d}f\arrI{d}z), \\
    \beta &= (v'\arrO{a}g\arrI{a}w'\arrO{b}h'\arrI{b} x'\arrO{cd} r'\arrI{cd}z')
  \end{align*}
  for some $v,x,y,z,v',w',x',z'\in X[\epsilon]$,
  $e\in X[c]$, $f\in X[d]$, $g'\in X[a]$, $h'\in X[b]$,
  $q\in X[\loset{a\\b}]$, $r'\in X[\loset{c\\d}]$.
  Furthermore, $x\neq x'$, and for
  \begin{align*}
    \bar\alpha &= (v\arrO{a}\delta^0_b(q)\arrI{a} \delta^0_a\delta^1_a(q)\arrO{b}\delta^1_a(q)\arrI{b} x\arrO{c}e\arrI{c}y\arrO{d}f\arrI{d}z), \\
    \bar\beta &= (v'\arrO{a}g\arrI{a}w'\arrO{b}h'\arrI{b} x'\arrO{c}\delta^0_d(r')\arrI{c} \delta^0_d\delta^1_c(r')\arrO{d}\delta^1_c(r')\arrI{d}z')
  \end{align*}
  we have $\ev(\bar{\alpha})=\ev(\bar\beta)=abcd$ and $\bar\alpha\neq\bar\beta$.
\end{lemma}

\begin{proof}
  The unique sparse step decomposition of $\loset{a\\b}cd$ is
  \[
    \loset{a\\b}cd=
    \loset{a\ibullet\\b\ibullet}
    *\loset{\ibullet a\\ \ibullet b}
    * [c\ibullet] * [\ibullet c] * [d\ibullet] * [\ibullet d].
  \]
  Thus, $\alpha$ must be as described above.
  A similar argument applies for $\beta$.

  Now assume that $x=x'$. Then
  \[
    \gamma=(v\arrO{ab}q\arrI{ab} x=x'\arrO{cd} r'\arrI{cd}z')
  \]
  is a path on $X$ for which $\ev(\gamma)=\loset{a\\b}*\loset{c\\d}$.
  Since $\gamma$ is essential,
  there are paths $\gamma'\in\Path(X)_\bot^v$ and $\gamma''\in\Path(X)_{z'}^\top$.
  The composition $\omega=\gamma' \gamma \gamma''$ is an accepting path.
  Thus, $\ev(\gamma')*\loset{a\\b}*\loset{c\\d}*\ev(\gamma'')\in L$: a contradiction.

  Calculation of $\ev(\bar\alpha)$ and $\ev(\bar\beta)$ is elementary,
  and $\bar\alpha\neq \bar\beta$ because $x\neq x'$. \qed
\end{proof}

\begin{proof}[of Proposition~\ref{pr:lambig}]
  Let $X$ be an HDA with $\Lang(X)=L$.
  We will show that there exist at least $2^n$ different sparse accepting paths
  accepting $(abcd)^n$.
  Let $P=\loset{a\\b}cd$ and $Q=ab\loset{c\\d}$.
  For every sequence $\mathbf{R}=(R_1,\dotsc,R_n)\in \{P,Q\}^n$
  let $\omega_{\mathbf{R}}$ be an accepting path
  such that $\ev(\omega_{\mathbf{R}})=R_1*\dotsm*R_n$.
  By Lemma~\ref{l:GeneralPathDivision},
  there exist paths $\omega_{\mathbf{R}}^1,\dotsc,\omega_{\mathbf{R}}^n$
  such that $\ev(\omega_{\mathbf{R}}^k)=R_k$
  and $\omega'_{\mathbf{R}}=\omega_{\mathbf{R}}^1*\dotsm*\omega_{\mathbf{R}}^n$
  is an accepting path.
  Let $\bar{\omega}_{\mathbf{R}}^k$ be the path defined as in Lemma~\ref{l:BrickPaths}
  (\ie like $\bar\alpha$ if $R_k=P$ and $\bar\beta$ if $R_k=Q$).
  Finally, put
  $
  \bar\omega_{\mathbf{R}}
  =\bar\omega_{\mathbf{R}}^1*\dotsm*\bar\omega_{\mathbf{R}}^n.
  $

  Now choose ${\mathbf{R}}\neq \mathbf{S}\in\{P,Q\}^n$.
  Assume that $\bar\omega_{\mathbf{R}}=\bar\omega_{\mathbf{S}}$.
  This implies that $\bar\omega_{\mathbf{R}}^k=\bar\omega_{\mathbf{S}}^k$
  for all $k$ (all segments have the same length).
  But there exists $k$ such that $R_k\neq S_k$ (say $R_k=P$ and $S_k=Q$),
  and, by Lemma \ref{l:BrickPaths} again, applied to $\alpha=\bar\omega_{\mathbf{R}}^k$
  and $\beta=\bar\omega_{\mathbf{S}}^k$, we get $\bar\omega_{\mathbf{R}}^k\neq \bar\omega_{\mathbf{S}}^k$: a contradiction.

  As a consequence, the paths $\{\bar\omega_{\mathbf{R}}\}_{\mathbf{R}\in\{P,Q\}^n}$
  are sparse and pairwise different, and $\ev(\bar\omega_{\mathbf{R}})=(abcd)^n$ for all $\mathbf{R}$. \qed
\end{proof}

\section{One-letter languages}
\label{se:one}

In this section we restrict our attention to languages on one letter only, with $\Sigma=\{a\}$.
We develop a characterization for deterministic regular one-letter languages
using ultimately periodic functions,
similar to the one of Büchi \cite{journals/mathgrund/Buchi60}.
Yet first we show that also in the one-letter case, not all languages are deterministic.

\begin{proposition}
  The regular language $(\loset{a\\a} a + a \loset{a\\a})^+$
  is of unbounded ambiguity and hence non-deterministic.
\end{proposition}

\begin{proof}
  A simple modification of the proof of Proposition \ref{pr:lambig},
  replacing the letters $b$, $c$ and $d$ by $a$. \qed
\end{proof}

Now let $X$ be a deterministic HDA on $\Sigma=\{a\}$.
Noting that $\square=\{a^{\parallel k}\mid k\ge 0\}$,
we write $X[k]$ for $X[\{a^{\parallel k}\}]$.

We assume that $\bot_X=\{v_0\}\subseteq X_0$ contains only a vertex;
the more general case when $X$ also has higher-dimensional start cells is similar and will not be treated below.
We also assume that $X$ is \emph{accessible}, that is, every cell is a target of a path starting at $v_0$.

By Lemma \ref{le:detskel}, $X_{\le 1}$ is a deterministic finite automaton.
Potentially adding an extra sink vertex,
we may further assume that $X_{\le 1}$ is \emph{complete}
in the sense that very vertex has an outgoing transition.

Define a function $\suc: X[0]\to X[0]$
by $\suc=\{(v, v')\mid \exists e\in X[1]: \delta^0_1(e)=v, \delta^1_1(e)=v'\}$.
This is indeed a function as $X_{\le 1}$ is deterministic and complete;
it is the same as the transition function of $X_{\le 1}$ seen as a deterministic finite automaton.

Define an infinite sequence $(v_0, v_1,\dotsc)$ by $v_{n+1}=\suc(v_n)$,
then $X[0]=\{v_0, v_1,\dotsc\}$ by accessibility.
The following is easy to see.

\begin{lemma}
  \label{l:DetermVert}
  There exist unique $r=r(X)\ge 1$ and $s=s(X)\ge 0$ such that
  \begin{itemize}
  \item
    $v_k=v_{k+r}$ for all $k\ge s$,
  \item
    $v_k\neq v_l$ for all $k<s\le l$. \qed
  \end{itemize}
\end{lemma}

\begin{figure}[tbp]
  \centering
  \begin{tikzpicture}[scale=0.95, every node/.style={transform shape}]
    \tikzstyle{corner}=[draw, circle, fill = white,text width=-3mm];
    \draw[white, fill = gray, fill opacity=0.2] (0,0) -- (0,1.5) -- (3,1.5) -- (3,0) --cycle;
    \draw[white, fill = gray, fill opacity=0.2] (3,1.5) -- (3,3) -- (4.5,3) -- (4.5,1.5) --cycle;
    \node[corner] (-1) at (-1.5,0) {};
    \node[corner] (0) at (0,0) {};
    \node[corner] (1) at (1.5,0) {};
    \node[corner] (2) at (3,0) {};
    \node[corner] (3) at (0,1.5) {};
    \node[corner] (4) at (1.5,1.5) {};
    \node[corner] (5) at (3,1.5) {};
    \node[corner] (6) at (4.5,1.5) {};
    \node[corner] (7) at (3,3) {};
    \node[corner] (8) at (4.5,3) {};
    \node[corner] (9) at (6,3) {};
    \node[corner] (10) at (7.5,3) {};
    
    \path
    (-1) edge (0)
    (0) edge (1) (0) edge (3)
    (1) edge (2) (1) edge (4)
    (2) edge (5)
    (3) edge (4)
    (4) edge (5)
    (5) edge (6) (5) edge (7)
    (6) edge (8)
    (7) edge (8)
    (8) edge (9)
    (9) edge (10)
    (10) edge[loop right] (10);
    
    \node (bot) at (-1.9,0) {$\bot$};
    \node (top) at (6,3.4) {$\top$};

    \path[-, very thick, bend right, dotted] (.75,0) edge (0,.75);
    \path[-, very thick, bend right, dotted] (2.25,0) edge (1.5,.75);
    \path[-, very thick, bend left, dotted] (1.5,.75) edge (.75,1.5);
    \path[-, very thick, bend right, dotted] (3.75,1.5) edge (3,2.25);
    \path[-, very thick, bend left, dotted] (3,.75) edge (2.25,1.5);
    \path[-, very thick, bend left, dotted] (4.5,2.25) edge (3.75,3);

    \node at (-1.5,-.3) {$v_0$};
    \node at (0,-.3) {$v_1$};
    \node at (1.5,-.3) {$v_2$};
    \node at (0,1.8) {$v_2$};
    \node at (3,-.3) {$v_3$};
    \node at (1.5,1.8) {$v_3$};
    \node at (3.3,1.2) {$v_4$};
    \node at (4.5,1.2) {$v_5$};
    \node at (3,3.3) {$v_5$};
    \node at (4.5,3.3) {$v_6$};
    \node at (6,2.7) {$v_7$};
    \node at (7.6,2.7) {$v_8, v_9,...$};
  \end{tikzpicture}
  \caption{%
    Deterministic HDA on one-letter alphabet (all edges labeled $a$).
    Edges connected with dotted lines are identified.
    Vertex labels indicate sequence $(v_0, v_1,\dotsc)$ given by $\suc$ function.}
  \label{fi:dethda}
\end{figure}

\begin{example}
  \label{ex:dethda}
  Figure \ref{fi:dethda} shows a simple deterministic HDA.
  Edges connected with dotted lines are identified, as are their corresponding end points.
  Also shown is the sequence $(v_0, v_1,\dotsc)$; in this case, $r=1$ and $s=8$.
\end{example}

Now define a function $f_X: \Nat\to \Nat$ by
\begin{equation*}
  f_X(n)=\max\{k\mid \exists x\in X[k]: \delta^0_{\ev(x)}(x)=v_n\}.
\end{equation*}
This function marks the maximal number of concurrent events that can be started in each vertex.
That is, $f_X(n)=k$ if there is a $k$-cell starting at $v_n$ and no higher dimensional cells start at $v_n$.
We collect some properties of $f_X$.

\begin{lemma}
  \label{l:DetermFun}
  For all $n\ge 0$,
  \begin{enumerate}[(1)]
  \item
    $f_X(n)\ge 1$,
  \item
    $f_X(n+1)\ge f_X(n)-1$, and
  \item
    $f_X(n)=f_X(n+r)$ for $n\ge s$.
  \end{enumerate}
\end{lemma}

\begin{proof}
  Every $v_n$ has a successor, which implies (1).
  If $x\in X[f(n)]$ is a cell such that $\delta^0_{\ev(x)}(x)=v_n$,
  then for every $b\in \ev(x)$ we have $\delta^0_{\ev(x)}(\delta^1_b(x))=v_{n+1}$, which implies (2).
  Finally, (3) follows from Lemma \ref{l:DetermVert}. \qed
\end{proof}

\begin{example}
  For the HDA of Example \ref{ex:dethda},
  the values of $f_X$ are
  \begin{equation*}
    (1, 2, 2, 1, 2, 1, 1, 1, 1,\dotsc).
  \end{equation*}
\end{example}

\begin{proposition}
  \label{p:DetPCS}
  Let $r\ge 1$, $s\ge 0$ and $f:\Nat\to\Nat$ be a function satisfying the conditions of Lemma \ref{l:DetermFun}.
  Then there exists a deterministic accessible HDA $X$ with a single start vertex such that $f=f_X$.
  Moreover, the underlying precubical set is unique.
\end{proposition}

\begin{proof}
  For $k\ge 0$ define
  \begin{gather*}
    X[k] = \{x^k_n\mid k\le f(n), n<s+r \},\\
    \delta^0_A(x^k_n) = x^{k-|A|}_n,\qquad
    \delta^1_A(x^k_n) = x^{k-|A|}_{n+|A|},
  \end{gather*}
  with the convention that $x^k_n=x^k_{n-r\lceil(n-s-r+1)/r\rceil  }$ for $n\ge s+r$.
  Let $\bot_X=\{x_0^0\}$.
  \begin{itemize}
  \item
    Determinism: If $A\subseteq U\in\square$ and $x_n^k\in X[U-A]$,
    then $k=|U-A|$ and $x^{|U|}_n$ is the only cell such that $\delta^0_{A}(x^{|U|}_n)=x^k_n$.	
  \item
    Accessibility: All vertices $x^0_n$ are accessible from $x_0^0$ and thus all cells $x^k_n$.
  \item
    Uniqueness: $X_{\le 1}$ is defined by $s$ and $r$ and hence unique.
    By determinism, for every $n$ there is at most one cell $x$ of a given dimension such that $\delta^0_{\ev(x)}(x)=x^0_n$:
    uniqueness of $X$ follows. \qed
  \end{itemize}
\end{proof}

Finally we obtain a classification of single-letter languages.

\begin{definition}
  Let $S$ be a set.
  A function $f:\Nat\to S$ is \emph{ultimately periodic}
  if there exists a pair $(r, s)\in\Nat_{\ge 1}\times \Nat$,
  called a \emph{period} of $f$,
  such that $f(n+r)=f(n)$ for every $n\ge s$.
\end{definition}

If $(r,s)$ is a period of $f$,
then for every $k\ge 0$ and $c\ge 1$,
$(cr, s+k)$ is a period as well.
Every ultimately periodic function admits a unique minimal period.

\begin{proposition}
  \label{p:DetOneLetUlPer}
  Let $\TD=\{(k,T)\mid k\in\Nat_{\ge 1}, T\subseteq \{0,\dotsc,k\}\}$.
  There is a 1-1 correspondence between
  \begin{itemize}
  \item
    deterministic languages on one letter containing ipomsets with empty start interface,
  \item
    ultimately periodic functions $\phi=(f,\tau):\Nat\to \TD$
    such that $f(n+1)\ge f(n)-1$.
  \end{itemize}
\end{proposition}

\begin{proof}
  Let $L$ be a deterministic language and let $X$ a deterministic HDA with $L=\Lang(X)$.
  By Proposition \ref{p:DetPCS},
  $X$ regarded as a precubical set is determined by an ultimately periodic function $f_X:\Nat\to\Nat_{\ge 1}$
  such that $f_X(n+1)\ge f_X(n)-1$.
  The sets $\tau(n)$ are defined by the condition that $k\in \tau(n)$ iff $x^k_n$ is accepting.
	
  For the converse, let $\phi: \Nat\to \TD$ be ultimately periodic and choose a period $(r,s)$.
  Let $X$ be the HDA defined in Proposition \ref{p:DetPCS},
  with $\top_X=\{x^k_n\mid k\in \tau(n)\}$,
  then $X$ is deterministic. \qed
\end{proof}

\section{Complement}
\label{se:comp}

The \emph{complement} of a language $L\subseteq \iiPoms$,
\ie $\iiPoms-L$, is generally not down-closed
and thus not a language.
If we define $\bcomp{}{L}=(\iiPoms-L)\down$,
then $\bcomp{}{L}$ \emph{is} a language,
but a \emph{pseudocomplement} rather than a complement:
because of down-closure, $L\cap \bcomp{}{L}=\emptyset$ is now false in general.
The following additional problem poses itself.

\begin{proposition}
  If $L$ is regular, then $\bcomp{}{L}$ has infinite width, hence is not regular.
\end{proposition}

\begin{proof}
  By Lemma~\ref{lem:finitewidth}, $\wid(L)$ is finite.
  For any $k>\wid(L)$, $\iiPoms-L$ contains all ipomsets of width $k$,
  hence $\{\wid(P)\mid P\in \bcomp{}{L}\}$ is unbounded.
  \qed
\end{proof}

To remedy this problem,
we introduce a width-bounded version of (pseudo) complement.
We fix an integer $k\ge 0$ for the rest of the section.

\begin{definition}
  \label{de:bcomp}
  The \emph{$k$-bounded complement} of $L\in \Langs$
  is
  $\bcomp{k}{L} = (\iiPoms_{\le k}- L)\down$.
\end{definition}

\begin{lemma}
  \label{le:bcomp_elm}
  Let $L$ and $M$ be languages.
  \begin{enumerate}
  \item $\bcomp{0}{L}=\{ \id_\emptyset \}-L$.
  \item $L\subseteq M$ implies $\bcomp{k}{M}\subseteq \bcomp{k}{L}$.
  \item \label{le:bcomp_elm.dbl} $\bcomp{k}{\bcomp{k}{L}}\subseteq L_{\le k}\subseteq L$.
  \item \label{le:bcomp_elm.krestrict} $\bcomp{k}{L} = \bcomp{k}{L_{\le k}}$
  \end{enumerate}
\end{lemma}

\begin{proof}
  \mbox{}
  \begin{enumerate}
  \item $\bcomp{0}{L}=(\iiPoms_{\le 0}-L)\down=\{\epsilon\}-L$.
  \item $L\subseteq M$ implies $\iiPoms_k-M\subseteq \iiPoms_k-L$, thus
    $(\iiPoms_k-M)\down \subseteq (\iiPoms_k-L)\down$.
  \item We have $\iiPoms_{\le k}-L_{\le k}=\iiPoms_{\le k}-L\subseteq  \bcomp{k}{L}$, so by the previous item,
    $\iiPoms_{\le k}-\bcomp{k}{L}\subseteq \iiPoms_{\le k} - (\iiPoms_{\le k}-L_{\le k})=L_{\le k}$.
    Thus, $\bcomp{k}{\bcomp{k}{L}}\subseteq L_{\le k}\down=L_{\le k}$.
  \item $\bcomp{k}{L} = (\iiPoms_{\le k} - L) \down = (\iiPoms_{\le k} - L_{\le k}) \down = \bcomp{k}{L_{\le k}}$.
  \end{enumerate}
  \qed
\end{proof}

\begin{proposition}
  For any $k\ge 0$, $\bcomp{k}{\,\cdot\,\vphantom{L}}$ is a pseudocomplement on the lattice $(\Langs_{\le k}, {\supseteq})$,
  that is, for any $L, M\in \Langs_{\le k}$,
  $L\cup M=\iiPoms_k$ iff $\bcomp{k}{L}\subseteq M$.
\end{proposition}

\begin{proof}
  Let $L, M\in \Langs_{\le k}$ such that $L\cup M=\iiPoms_k$ and $P\in \bcomp{k}{L}$.
  There exists $Q\in \iiPoms_{\le k}$ such that $P\subsu Q$ and $Q \not\in L$.
  Thus, $Q\in M$ and since $M$ is closed by subsumption, $P\in M$.

  Conversely, let $L, M\in \Langs_{\le k}$ such that $\bcomp{k}{L}\subseteq M$ and $P\in \iiPoms_{\le k}-M$.
  Then $P\in \bcomp{k}{M}$, and we have that $\bcomp{k}{M} \subseteq \bcomp{k}{\bcomp{k}{L}} \subseteq L$
  by Lemma~\ref{le:bcomp_elm}(\ref{le:bcomp_elm.dbl}).
  Thus, $P\in L$ and then $L\cup M=\iiPoms_k$.
  \qed
\end{proof}

The pseudocomplement property immediately gets us the following.

\begin{corollary}
  \label{cor:bcomp}
  Let $k\ge 0$ and $L, M\in \Langs_{\le k}$.  Then
  $L\cup \bcomp{k}{L}=\iiPoms_{\le k}$,
  $\smash[t]{\bcomp{k}{\bcomp{k}{\bcomp{k}{L}}}}=\bcomp{k}{L}$,
  $\bcomp{k}{L\cap M}=\bcomp{k}{L}\cup \bcomp{k}{M}$,
  $\bcomp{k}{L\cup M}  \subseteq \bcomp{k}{L} \cap \bcomp{k}{M}$, and
  $\bcomp{k}{\bcomp{k}{L\cup M}}=\bcomp{k}{\bcomp{k}{L}}\cup \bcomp{k}{\bcomp{k}{M}}$.
  Further, $\bcomp{k}{L}=\emptyset$ iff $L=\iiPoms_{\le k}$.
\end{corollary}

For $k=0$ and $k=1$, $\bcomp{k}{\phantom{L}}$ is a complement on $\iiPoms_{\le k}$,
but for $k\ge 2$ it is not:
in general, neither $L=\bcomp{k}{\bcomp{k}{L}}$, $L\cap \bcomp{k}{L}=\emptyset$,
nor $\bcomp{k}{L\cup M}=\bcomp{k}{L}\cap \bcomp{k}{M}$ hold:

\begin{figure}[tbp]
  \centering
  \begin{tikzpicture}[scale=0.9, every node/.style={transform shape}]
    \tikzstyle{corner}=[draw, circle, fill = white,text width=-3mm];
    \draw[fill = gray, fill opacity=0.2] (0,0) -- (0,1.5) -- (3,1.5) -- (3,0) --cycle;
    \node[corner] (0) at (0,0) {};
    \node[corner] (1) at (1.5,0) {};
    \node[corner] (2) at (3,0) {};
    \node[corner] (3) at (0,1.5) {};
    \node[corner] (4) at (1.5,1.5) {};
    \node[corner] (5) at (3,1.5) {};
    
    \draw[-stealth] (3) to node[midway, above]{$a$} (4);
    \draw[-stealth] (4) to node[midway, above]{$b$} (5);
    \draw[-stealth] (0) to node[midway, left]{$c$} (3);
    
    \draw[-stealth] (0) to node[midway, below]{$\vphantom{b}a$} (1);
    \draw[-stealth] (1) to node[midway, below]{$b$} (2);
    \draw[-stealth] (1) to node[midway, left]{$c$} (4);
    
    \draw[-stealth] (2) to node[midway, right]{$c$} (5);
    
    \node[] (bot1) at (-0.4,-0.1) {$\bot$};
    \node[] (bot1) at (3.4,1.6) {$\top$};
    \begin{scope}[shift={(4.2,0)}]
      \tikzstyle{corner}=[draw, circle, fill = white,text width=-3mm];
      
      \draw[fill = gray, fill opacity=0.2] (0,0) -- (1.5,0) -- (1.5,3) -- (0,3) --cycle;
      
      \node[corner] (0) at (0,0) {};
      \node[corner] (1) at (0,1.5) {};
      \node[corner] (2) at (0,3) {};
      \node[corner] (3) at (1.5,0) {};
      \node[corner] (4) at (1.5,1.5) {};
      \node[corner] (5) at (1.5,3) {};
      
      \draw[-stealth] (3) to node[midway, right]{$a$} (4);
      \draw[-stealth] (4) to node[midway, right]{$b$} (5);
      \draw[-stealth] (0) to node[midway, below]{$c$} (3);
      
      \draw[-stealth] (0) to node[midway, left]{$a$} (1);
      \draw[-stealth] (1) to node[midway, left]{$b$} (2);
      \draw[-stealth] (1) to node[midway, above]{$c$} (4);
      
      \draw[-stealth] (2) to node[midway, above]{$c$} (5);
      
      \node[] (bot1) at (-0.4,-0.1) {$\bot$};
      \node[] (bot1) at (1.9,3.1) {$\top$};
    \end{scope}
    \begin{scope}[shift={(7.5,1.5)}]
      \node (a) at (0.4,0.7) {$a$};
      \node (c) at (0.4,-0.7) {$c$};
      \node (b) at (1.8,0) {$b$};
      \path (a) edge (b);
      \path[densely dashed, gray]  (b) edge (c);
      \path[densely dashed, gray]  (a) edge (c);
      
      \node (a2) at (3.4,-0.7) {$a$};
      \node (c2) at (3.4,0.7) {$c$};
      \node (b2) at (4.8,0) {$b$};
      \path (a2) edge (b2);
      \path[densely dashed, gray]  (c2) edge (b2);
      \path[densely dashed, gray]  (c2) edge (a2);
    \end{scope}
  \end{tikzpicture}
  \caption{HDA $X$ (disconnected) which accepts language $L$ of Example~\ref{ex:compnotcup}
    and the two generating ipomsets in $L$.}
  \label{fig:countercompunion}
\end{figure}

\begin{example}
  \label{ex:compnotcup}
  Let $A=\{P\in \iiPoms_{\le 2} \mid abc \subsu P \}$, $L=\{\loset{a\to b\\c}, \loset{c\\a\to b}\}\down$
  and $M=(A - L)\down$.
  The HDA $X$ in Figure~\ref{fig:countercompunion} accepts $L$.
  Notice that due to the non-commutativity of parallel composition (because of event order),
  $X$ consists of two parts, one a ``transposition'' of the other.
  The left part accepts $\loset{a\to b\\c}$, while the right part accepts $\loset{c\\a\to b}$.

  Now $a b c\subsu \loset{a\to c\\b}$ which is not in $L$, so that $a b c\in \bcomp{2}{L}$.
  Similarly, $a b c\subsu \loset{a\to b\\c}\notin M$, so $a b c\in \bcomp{2}{M}$.
  Thus, $a b c\in \bcomp{2}{L}\cap \bcomp{2}{M}$.
  On the other hand,
  for any $P$ such that $\wid(P)\le 2$ and $abc \subsu P$, we have $P \in L\cup M=A\down$.
  Hence $a b c\notin \bcomp{2}{L\cup M}$.
  
  Finally,  $\bcomp{3}{L}$ contains every ipomset of width $3$,
  hence $\bcomp{3}{L}=\iiPoms_{\le 3}$,
  so that $L\cap \bcomp{3}{L}=L\ne \emptyset$
  and $\bcomp{3}{\bcomp{3}{L}}=\emptyset\ne L$.
  This may be generalised to the fact that $\smash[t]{\bcomp{k}{\bcomp{k}{L}}}=\emptyset$
  as soon as $\wid(L)<k$.
\end{example}

For $k\ge 0$, let $\skeletal_k$ be the set of all languages $L$ for which $L = \bcomp{k}{\bcomp{k}{L}}$.
We characterise $\skeletal_k$ in the following.
By $\bcomp{k}{\bcomp{k}{\bcomp{k}{L}}}=\bcomp{k}{L}$ (Cor.~\ref{cor:bcomp}),
$\skeletal_k=\{\bcomp{k}{L}\mid L\in \Langs\}$,
\ie $\skeletal_k$ is the image of $\Langs$ under $\bcomp{k}{\phantom{L}}$.
(This is a general property of pseudocomplements.)

Define
$\maxSubsume_k =
  \{P \in \iiPoms_{\le k} \mid \forall Q \in \iiPoms_{\le k}: Q\neq P\implies P \not\subsu Q \}$,
the set of all $\subsu$-maximal elements of $\iiPoms_{\le k}$.
In particular, $\maxSubsume_k\down = \iiPoms_{\le k}$.
Note that $P\in \maxSubsume_k$ does not imply $\wid(P) = k$:
for example, $\loset{a\\b}\in \maxSubsume_3$.

\begin{lemma}
  \label{le:downmax}
  For any $L\in\Langs$,
  $\bcomp{k}{L}=(\maxSubsume_k - L)\down$.
\end{lemma}

\begin{proof}
  We have
  \begin{align*}
    Q\in \bcomp{k}{L}
    &\iff
    \exists P\in (\iiPoms_{\le k}- L):\; Q\subsu P\\
    &\iff
    \exists P\in (\iiPoms_{\le k}- L)\cap \maxSubsume_k:\; Q\subsu P\\
    &\iff
    \exists P\in \maxSubsume_k-L:\; Q\subsu P
    \iff
    Q\in (\maxSubsume_k - L)\down.
  \end{align*}

  \vspace*{-5ex}\qed
\end{proof}

\begin{corollary}
  \label{co:maxsubsume}
  Let $L \in \Langs$ and $k \ge 0$,
  then
  $\bcomp{k}{L} = \iiPoms_{\le k}$ iff $L \cap \maxSubsume_k = \emptyset$.
\end{corollary}

\begin{proposition}
  $\skeletal_k=\{A\down\mid A\subseteq \maxSubsume_k\}$.
\end{proposition}

\begin{proof}
  Inclusion $\subseteq$ follows from Lemma~\ref{le:downmax}.
  For the other direction,
  $A\subseteq \maxSubsume_k$ implies
  \begin{equation*}
    \bcomp{k}{\bcomp{k}{A\down\,}}
    = \bcomp{k}{(\maxSubsume_k-A\down)\down}
    = \bcomp{k}{(\maxSubsume_k-A)\down}
    = (\maxSubsume_k-(\maxSubsume_k-A))\down
    = A\down.
  \end{equation*}

  \vspace*{-5ex}\qed
\end{proof}

If $A\neq B\subseteq \maxSubsume_k$,
then also $A\down\neq B\down$,
since all elements of $\maxSubsume_k$ are $\subsu$-maximal.
As a consequence,
$\skeletal_k$ and the powerset $\mathcal{P}(\maxSubsume_k)$ are isomorphic lattices,
hence $\skeletal_k$ is a distributive lattice with join $L\vee M=L\cup M$
and meet $L\wedge M=(L\cap M\cap \maxSubsume_k)\down$.

\begin{corollary}
  For $L, M\in \Langs$, $\bcomp{k}{L} = \bcomp{k}{M}$ iff $L\cap \maxSubsume_k=M\cap \maxSubsume_k$.
\end{corollary}

We can now show that bounded complement preserves regularity.

\begin{theorem}
  \label{th:compClosure}
  If $L\in \Langs$ is regular, then for all $k\ge 0$ so is $\bcomp{k}{L}$.
\end{theorem}

\begin{proof}
  By Proposition~\ref{prop:regrestriction}, $L_{\le k}$ is regular.
  Let $X$ be an HDA such that $\Lang(X) = L_{\le k}$ and $k=\dim(X)$.
  The $\Omega_{\le k}$-language $\Cohnew\cap \Phi(\Lang(X))$
  is regular by Lemma~\ref{le:philang},
  hence so is $\Cohnew-\Phi(\Lang(X))$.
  By Lemma~\ref{le:psireg}, $\Psi$ preserves regularity,
  so $\Psi(\Cohnew-\Phi(\Lang(X)))$ is a regular ipomset language.
  Now for $P\in \iiPoms_{\le k}$ we have
  \begin{align*}
    P &\in \Psi(\Cohnew-\Phi(L_{\le k})) \\
    &\iff
    \exists Q \sqsupseteq P, \exists Q_1\dotsm Q_n\in\Cohnew - \Phi(L_{\le k})
    : Q=Q_1*\dotsm* Q_n
    \\
    &\iff
    \exists Q_1\dotsm Q_n\in\Cohnew
    : P\subsu Q_1*\dotsm*Q_n\not\in L_{\le k}
    \\ &\iff P\in \bcomp{k}{L_{\le k}},
  \end{align*}
  hence $\bcomp{k}{L_{\le k}}=\Psi(\Cohnew-\Phi(L_{\le k}))$.
  Lemma~\ref{le:bcomp_elm}(\ref{le:bcomp_elm.krestrict}) allows us to conclude.
  \qed
\end{proof}

\begin{corollary}
  $\iiPoms_{\le k}$ is regular for every $k\ge 0$.
\end{corollary}

\section{Conclusion and further work}

We have advanced the theory of higher-dimensional automata (HDAs) along several lines:
we have
shown a pumping lemma,
exposed a regular language of unbounded ambiguity,
introduced width-bounded complement,
shown that regular languages are closed under intersection and width-bounded complement,
and shown that determinism of a regular language and inclusion of regular languages are decidable.

A question which is still open is if it is decidable
whether a regular language is of bounded ambiguity.
On a more general level, a notion of recognizability is still missing.
This is complicated by the fact that
ipomsets do not form a monoid but rather a 2-category with lax tensor~\cite{DBLP:journals/iandc/FahrenbergJSZ22}.

Even more generally, a theory of weighted and/or timed HDAs would be called for,
with a corresponding Kleene-Schützenberger theorem.
For timed HDAs, some work is available in~\cite{DBLP:journals/lites/Fahrenberg22, DBLP:conf/apn/AmraneBCF24}.
For weighted HDAs, the convolution algebras of~\cite{journals/alguni/FahrenbergJSZ23} may provide a useful framework.
\cite{conf/ramics/AmraneBCFZ24} shows a close relation between ipomsets and non-free category generated by starters and terminators,
and we are wondering whether this could provide a path both to a treatment of recognizability and to weighted HDAs.

\bibliographystyle{plain}
\bibliography{mybib}

\end{document}